\documentclass[11pt]{amsart}
\usepackage{amsmath}
\usepackage{amsfonts}
\usepackage{amssymb,color}
\usepackage[hidelinks]{hyperref}
\usepackage{amsthm,doi}
\usepackage{thmtools}
\usepackage{tikz}
\usepackage{pgfplots}
\pgfplotsset{compat=1.16}
\allowdisplaybreaks
\usepackage{enumitem}

\numberwithin{equation}{section}

\newcommand{\R}{\mathbb{R}}

\newcommand{\N}{\mathbb{N}}
\newcommand{\C}{\mathbb{C}}
\newcommand{\Z}{\mathbb{Z}}

\newcommand{\No}{\mathcal{N}}

\renewcommand{\S}{\mathcal{S}}

\newcommand{\E}{\mathbb{E}}
\renewcommand{\P}{\mathbb{P}}

\renewcommand{\ker}{{\sf Ker}\,}

\newcommand{\ud}{\,\mathrm{d}}
\newcommand{\norm}[2]{\left\| #2 \right\|_{#1}}

\newcommand{\abs}[1]{\ensuremath{\left| #1 \right| }}
\newcommand{\pbox}{\hfill$\Box$\\}

\newtheorem{lemma}{Lemma}[section]
\newtheorem{theorem}[lemma]{Theorem}

\newtheorem{prop}[lemma]{Proposition}

\newtheorem{rem}[lemma]{Remark}

\author{Jos\'e Luis Romero}
\address[J. L. Romero]{Faculty of Mathematics \\
	University of Vienna \\
	Oskar-Morgenstern-Platz 1 \\
	A-1090 Vienna, Austria \\and
	Acoustics Research Institute\\ Austrian Academy of Sciences\\Dr. Ignaz Seipel-Platz 2,	AT-1010 Vienna, Austria}
\email{jose.luis.romero@univie.ac.at}

\author{Michael Speckbacher}
\address[M. Speckbacher]{Faculty of Mathematics \\
	University of Vienna \\
	Oskar-Morgenstern-Platz 1 \\
	A-1090 Vienna, Austria }
\email{michael.speckbacher@univie.ac.at}

\thanks{J. L. R. and M. S. gratefully acknowledge support from the Austrian Science Fund (FWF): 10.55776/Y1199.}

\title[Estimation of binary masks from ambient noise]{Estimation of binary time-frequency masks from ambient noise}


\begin{document}

\begin{abstract}
We investigate the retrieval of a binary time-frequency mask from a few observations of filtered white ambient noise. Confirming household wisdom in acoustic modeling, we show that this is possible by inspecting the average spectrogram of ambient noise. Specifically, we show that the lower quantile of the average of $\mathcal{O}(\log(\abs{\Omega}/\varepsilon))$ masked spectrograms is enough to identify a rather general mask $\Omega$ with confidence at least $\varepsilon$, up to shape details concentrated near the boundary of $\Omega$. As an application, the expected measure of the estimation error is dominated by the perimeter of the time-frequency mask. The estimator  requires no knowledge of the noise variance, and only a very qualitative profile of the filtering window, but no exact knowledge of it.
\end{abstract}

\maketitle

\section{Introduction}
Time-frequency masking is the process of altering a signal $f: \mathbb{R}^d \to \mathbb{C}$ by weighting its time-frequency profile with a \emph{mask}. 
A standard model for such operation is based on the \emph{short-time Fourier transform} of $f$. First a profile of $f$ is built out of its time-frequency correlations with an auxiliary smooth and fast-decaying \emph{window function} $g:\mathbb{R}^d \to \mathbb{C}$,
\begin{align}\label{eq_stft}
V_gf(x,\xi)=\int_{\R^d}f(t)\overline{g(t-x)}e^{-2\pi i \xi\cdot t}\ud t, \quad (x,\xi)\in\R^{2d}.
\end{align}
Second, a mask $m: \mathbb{R}^d\times\mathbb{R}^d \to 
\mathbb{C}$ is applied to the time-frequency profile of $f$ by solving the optimization problem
\begin{align}\label{eq_Hm_extremal}
H_m f = \mathrm{argmin}_{\tilde{f}} \int_{\mathbb{R}^d \times \mathbb{R}^d}
\Big| m(x,\xi) \cdot V_gf(x,\xi) - V_g \tilde{f}(x,\xi) \Big|^2 \ud x \ud \xi,
\end{align}
whose solution is, more explicitly,
\begin{align}\label{eq_Hm}
H_m f(t) = \int_{\mathbb{R}^{d} \times \mathbb{R}^d} m(x,\xi) V_g f(x,\xi)
{g(t-x)} e^{2\pi i \xi\cdot t}\ud x \ud \xi, \qquad t \in \mathbb{R}^d,
\end{align}
 {provided that $g$ is normalized by
	$\|g\|_2=1$ ---	see Section \ref{sec_tf}}. While the filter \eqref{eq_Hm} depends on the auxiliary window function $g$, different choices of smooth and fast-decaying $g$ yield qualitatively similar results.

The subclass of \emph{binary masks}, $m(x,\xi) \in \{0,1\}$, plays a distinguished role in acoustic modeling --- see, e.g., \cite{wang2005ideal,wang2006computational}. In this case $m$ is the indicator function of a domain $\Omega \subseteq \mathbb{R}^d\times\mathbb{R}^d$ and \eqref{eq_Hm} models a \emph{hard filter} that rejects certain time-frequency components, for example, following a perceptual criterion. The \emph{direct estimation problem} consists in producing a binary mask optimally tailored to a certain task such as the separation of speech from noise in a particular auditory setup. In contrast, the \emph{inverse estimation problem} consists in learning the mask of a filter from (known or unknown) test inputs, for example, so as to retrieve the current calibration of a device that simulates or enhances human hearing \cite{hazrati2013blind}.

While the problem of retrieving a mask $m$ from observations of filtered signals $H_m f$ is challenging --- as it is one example of operator identification, see \cite{hlawatsch2011wireless, MR2191780,MR2300357,MR3048592,MR3447991,MR3158014} and the references therein --- learning a binary mask is expected to be substantially easier. In fact, it is household wisdom of acoustic modeling that a binary mask is revealed by simply inspecting the filtered output of \emph{ambient noise}. The purpose of this article is to provide a formal analysis of this fact in the idealized setting of \emph{white} ambient noise --- that is, noise whose stochastics are independent at disjoint spatial locations. Crucially, \emph{we do not assume the variance of the noise to be known}, and, similarly, we do not assume exact knowledge of the window function $g$. 

\section{Results}
\subsection{Mask estimation with unknown noise level}\label{sec_r1}
\subsubsection*{Setup} Let $g: \mathbb{R}^d \to \mathbb{C}$ be a Schwartz function normalized by $\norm{2}{g}=1$, let $\Omega \subseteq \mathbb{R}^{2d}$ be a compact domain, and denote $H_\Omega$ the filter \eqref{eq_Hm} with binary mask
\begin{align*}
m(z) = \begin{cases}
	1 & \mbox{if }z\in\Omega\\
	0 & \mbox{if }z\notin\Omega
\end{cases}.
\end{align*}

\subsubsection*{Ambient noise model} We consider zero-mean Gaussian \emph{complex white noise} $\mathcal{N}$ with variance $\sigma^2>0$, that is, a generalized complex-valued Gaussian stochastic process on $\mathbb{R}^d$ satisfying
\begin{align*}
&\mathbb{E}\big\{\mathcal{N}(z) \cdot \overline{\mathcal{N}(w)}\big\} = \sigma^2 \cdot \delta(z-w),
\\
&\mathbb{E}\big\{\mathcal{N}(z) \cdot \mathcal{N}(w)\big\} = 0.
\end{align*}
Measure-theoretic aspects are presented in more detail in Section \ref{sec_noise}. The key assumption is \emph{whiteness}: If $E_1, E_2 \subset \mathbb{R}^d$ are disjoint Borel sets, then $\mathcal{N}_{|E_1}$ and $\mathcal{N}_{|E_2}$ are statistically independent.  Equivalently, whiteness means that the noise has equal intensity at different frequencies, and thus a constant power spectral density. White noise is a mathematical idealization of empirically observable ambient noise, which has flat power spectrum over a certain range of relevant frequencies, determined by the propagation medium or generating mechanism.

The assumption that $\mathcal{N}$ is complex-valued is made for mathematical convenience, as it is compatible with the use of complex exponentials in the Fourier transform. (Indeed, the modulation operation
$\mathcal{N} \to e^{2\pi i \xi \cdot\ } \mathcal{N}$, $\xi \in \mathbb{R}^d$, preserves the stochastics of $\mathcal{N}$). The minor modifications necessary to treat real-valued white noise are sketched in Section \ref{sec_real}.

\subsubsection*{Data} Suppose that we observe $H_\Omega(\mathcal{N}_1), \ldots, H_\Omega(\mathcal{N}_K)$, where $\mathcal{N}_1, \ldots, \mathcal{N}_K$ are $K$ independent realizations of complex white noise with \emph{unknown variance} $\sigma^2>0$. 
Even though a typical realization of white noise $\mathcal{N}$ is not a function on $\mathbb{R}^d$ but merely a tempered distribution, the fact that $g \in \mathcal{S}(\R^d)$ allows us to interpret $H_\Omega(\mathcal{N})$ distributionally --- see Section \ref{sec_noise} for details.

{As a modeling assumption the use of (pure) white noise as input means that the effective bandwidth of the ambient noise is larger than the frequency diameter of the filtering mask $\Omega$, as the filter \eqref{eq_Hm} is largely insensitive to input frequencies beyond the frequency support of $\Omega$. In other words, effective noise bandwidth limits beyond the frequency support of the mask do not need to be modeled, as they are automatically enforced by the filter \eqref{eq_Hm}	--- see Section \ref{sec_d}. In practical terms, the input model thus requires that the ambient noise be perceived as white within a certain auditory setup, modeled by the filter \eqref{eq_Hm}.}

\subsubsection*{Goal} We wish to estimate $\Omega$ based on the available data.

 {To the best of our knowledge, this is the first work to rigorously study the retrieval of \emph{binary} time-frequency masks. The focus on such masks is motivated by certain applications, notably acoustics \cite{wang2005ideal,wang2006computational}, and leads to important differences with respect to the identification problem for general operators. For example, a standard assumption in the theory of \emph{operator sampling} \cite{hlawatsch2011wireless, MR2191780,MR2300357,MR3447991} is that the unknown operator has \emph{bandlimited Kohn-Nirenberg symbol} \cite{MR2589451,MR3048592,MR3158014}, which is not satisfied by the filters studied in this article.}

\subsubsection*{Estimation limits}  {To understand what kind of performance guarantees can be expected, let us look more closely into the masking procedure \eqref{eq_Hm_extremal}. The time-frequency filter with binary mask $m=\chi_\Omega$ finds the short-time Fourier transform $V_g H_\Omega f$ that best matches the function $\chi_\Omega V_g f$ in the quadratic mean sense. On the one hand, due to the nature of the windowing procedure defining the short-time Fourier transform, the function $V_g H_\Omega f$ is smooth, independently of the regularity of the input signal $f$ \cite[Chapter 3]{groe1} (which in our case is white noise). On the other hand, $\chi_\Omega V_g f$ has jump discontinuities along $\partial \Omega$. Hence, the operation $f \mapsto H_\Omega f$ has a smoothing effect along the time-frequency components located near the boundary of $\Omega$ (see Figure \ref{fig_proj}) and thus suppresses the fine details of the boundary of $\Omega$.}

\begin{figure}[ht]
	\begin{center}
		\includegraphics[height=5cm]{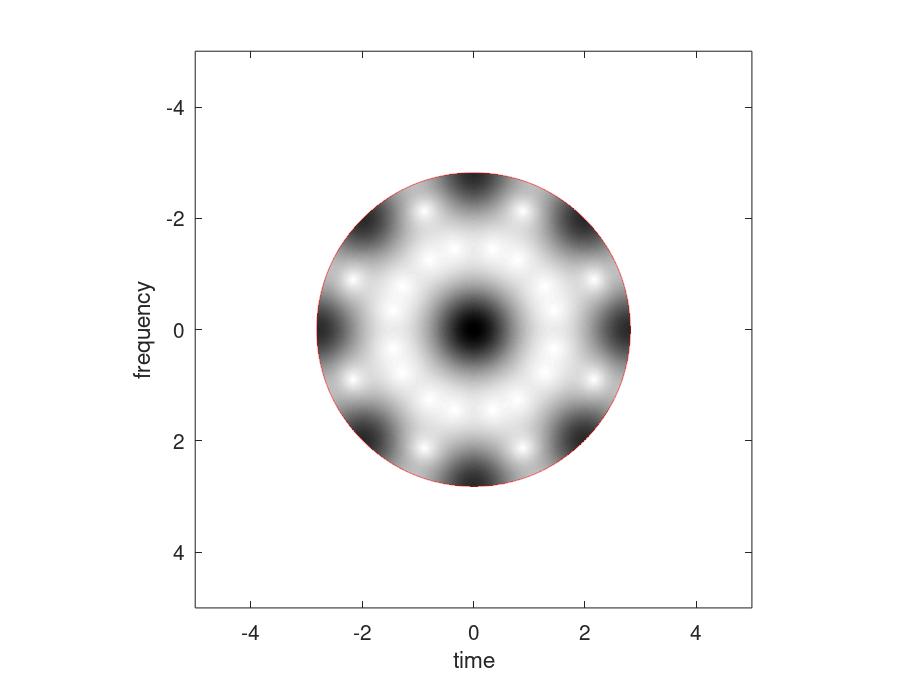}
		\includegraphics[height=5cm]{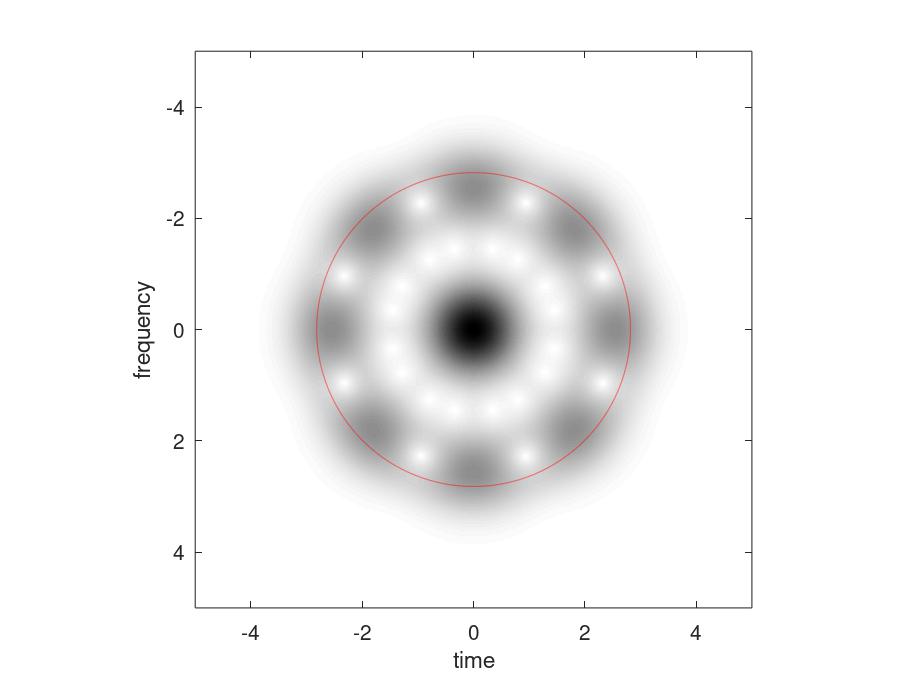}
		\\
		\includegraphics[height=5cm]{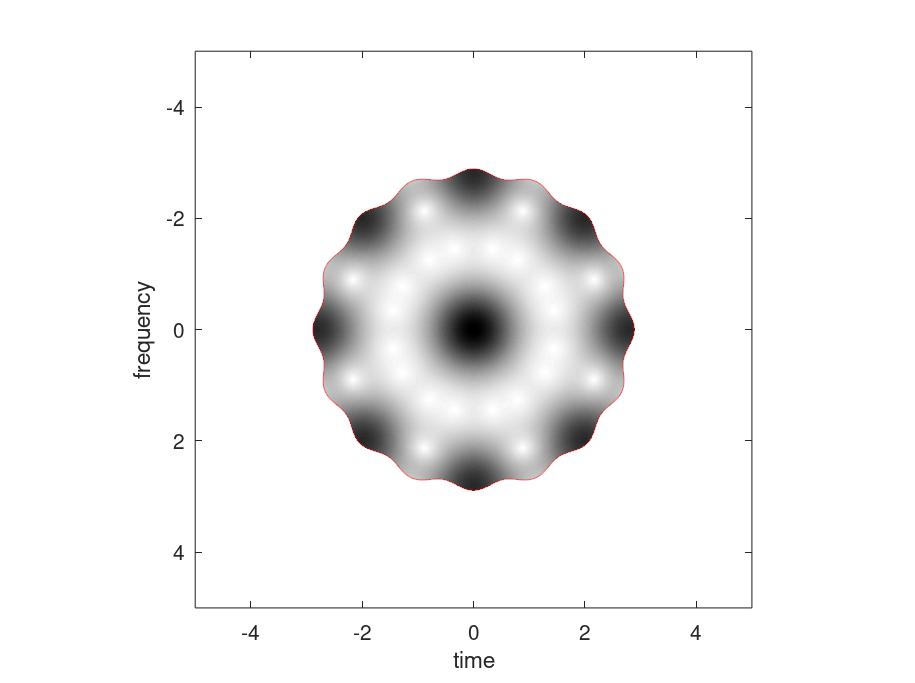}
		\includegraphics[height=5cm]{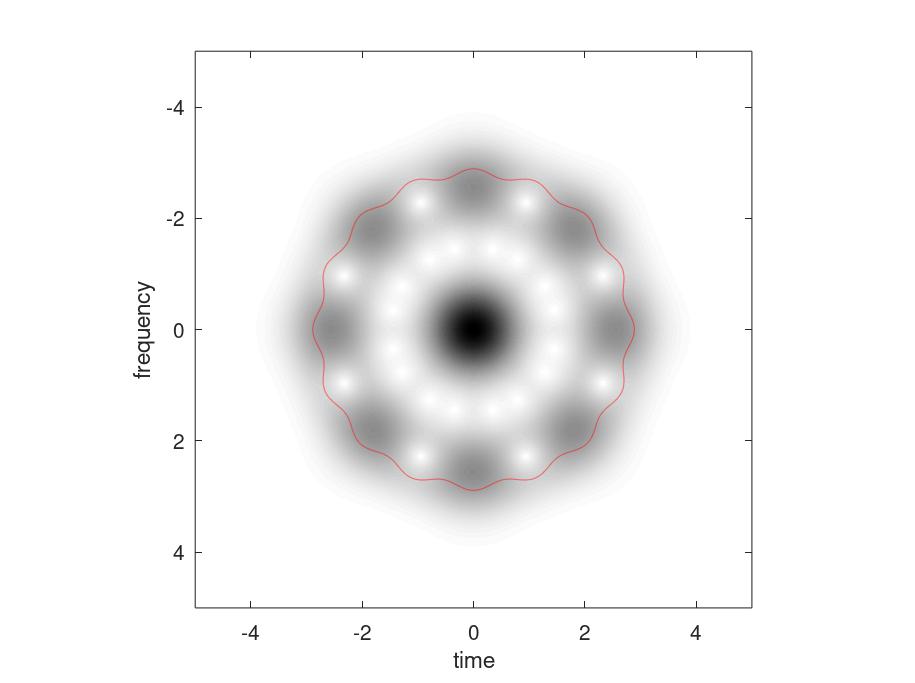}
		\caption{Time-frequency filters can suppress details along the mask boundary. Plot of $|V_\varphi f|\cdot \chi_\Omega$ (left) and $|V_\varphi 
			H_\Omega f|$ (right), with $\varphi(t)=2^{1/4}e^{-\pi t^2}$, $f$ a linear combination of time-frequency shifts of $\varphi$, and  $\Omega$ a disk (top) compared to  a slightly 
			perturbed  disk (bottom).}
		\label{fig_proj}
	\end{center}
\end{figure}

\subsubsection*{Proposed estimator} We first fix a function $\varphi \in \mathcal{S}(\mathbb{R}^d)$, which we call \emph{reconstruction window}, normalized by $\lVert\varphi\rVert_{2}=1$. A typical choice is a Gaussian $\varphi(t)=2^{d/4} e^{-\pi t^2}$. We then let
\begin{align}\label{eq_asp}
\rho(z):=\frac{1}{K}\sum_{k =1}^K \big|V_\varphi\big(H_\Omega \No_k\big)(z) \big|^2, \qquad
z=(x,\xi) \in \mathbb{R}^d\times\mathbb{R}^d
\end{align}
be the \emph{average observed spectrogram}. As an estimate of $\Omega$, we propose the following lower quantile of $\rho$:
\begin{align}\label{eq_estimator}
\widehat{\Omega} := \Big\{ z \in \mathbb{R}^d\times\mathbb{R}^d\,:\,
\rho(z) \geq \frac{1}{4} \max_{w \in \mathbb{R}^{2d}} \rho(w)
\Big\}.
\end{align}

\subsubsection*{Assumptions} 
We  assume that $\Omega$ is reasonably large in the sense that its measure dominates its perimeter:
\begin{align}\label{eq_as2}
|\Omega| \geq \max\Big\{ 2,  
8\int_{\R^{2d}}|V_gg(z)|^2|z|\ud z \cdot |\partial\Omega| \Big\}.
\end{align}
Here, $\abs{\partial \Omega}$ denotes the \emph{perimeter} of $\Omega$, that is, the $(2d-1)$-dimensional measure of its boundary (see Section \ref{sec_fp} for a detailed discussion).  {In particular, as $\Omega$ is compact, \eqref{eq_as2} requires the perimeter $|\partial\Omega|$ to be finite.}

 {To motivate the largeness condition \eqref{eq_as2}, recall the above discussion of estimation limits. Our assumption means that, in the proposed recovery problem, the measure of the estimation objective $\Omega$ is large with respect to the expected source of error $\partial\Omega$.}
	
 {In the literature, it is most common to analyze time-frequency filters \eqref{eq_Hm} by considering binary masks produced by dilation of a certain domain, $\Omega = r \cdot \Omega_0$, and their large-scale asymptotics, $r \to \infty$ \cite{MR1310645}. In this regime, the largeness condition \eqref{eq_as2} holds for all sufficiently large $r$, as long as $\Omega_0$ is compact and has finite perimeter,
	because $|\Omega|=r^{2d} |\Omega_0|$, while $|\partial\Omega|=r^{2d-1}|\partial\Omega_0|$. Our analysis is finer and less asymptotic, as it concerns a given domain and a concrete largeness condition, that is satisfied at a concrete scale, determined by the relation between the measure and perimeter of the masking domain $\Omega$ and the time-frequency concentration of the model window $g$.}

\subsubsection*{Performance guarantees} 
The \emph{estimation error} is the symmetric difference set
\begin{align}
	\Omega \triangle \widehat{\Omega} = \big(\Omega \setminus \widehat{\Omega}\big)
	\cup \big(\widehat{\Omega} \setminus \Omega\big).
\end{align}
Our main result shows that with overwhelming probability the estimation error is concentrated around the \emph{boundary} of $\Omega$; see Figure \ref{fig:1} for an illustration. More precisely, 
denoting the $r$-neighborhood of a set $E \subseteq \mathbb{R}^{2d}$ by
\begin{align*}
E^r := \big\{z \in \mathbb{R}^{2d}: d(z,E) < r\big\},
\end{align*}
the main result reads as follows.
\begin{theorem}[Recovery up to boundary details]
	\label{mthm1}
	Let $g,\varphi\in\S(\R^d)$ with $\norm{2}{g}=\norm{2}{\varphi}=1$ be model and reconstruction windows.
	Let $\Omega\subset \R^{2d}$ be compact and satisfying the largeness condition \eqref{eq_as2}. Let $K \geq 4$, assume that $K$ independent realizations of $H_\Omega(\mathcal{N})$ are observed, and consider the estimator
	$\widehat{\Omega}$ defined in \eqref{eq_estimator}.
	
	Then there exist positive constants $c, C, r>0$, that depend on $g$ and $\varphi$ but not on $\Omega$, such that with probability at least $1-c \cdot |\Omega^r| \cdot e^{-C K}$, the estimation error satisfies
	\begin{equation}\label{eq_rec}
		\Omega \triangle \widehat{\Omega} \subseteq \big[\partial\Omega\big]^r.
	\end{equation}
\end{theorem}
Let us remark on some aspects of Theorem \ref{mthm1}:

\noindent $\bullet$ The estimation method \emph{does not require knowledge of the noise level} $\sigma>0$. Similarly, the performance guarantees are independent of $\sigma$.

\noindent $\bullet$  {As discussed above, recovery up to a neighborhood of the boundary of $\Omega$ is intuitively the best possible result, as the filter $H_\Omega$ tends to suppress boundary details. Theorem \ref{mthm1} shows that given a confidence level $\varepsilon>0$ this is achieved with only $\mathcal{O}\big(\log(\abs{\Omega^r}/\varepsilon)\big)$ measurements of ambient noise.}

\noindent $\bullet$ While the estimation method does not require precise knowledge of the model window $g$, 
its time-frequency concentration does impact the performance guarantees through the largeness condition \eqref{eq_as2}.
Similarly, the correlation between the model and reconstruction windows affects the overall estimation performance.  {In fact, as we show in Remark \ref{rem_cons}, the dependence  of the constants $c,C,r$ on $g$ and $\varphi$ can be made explicit, and the main contribution is from} the following mutual correlation measure:
\begin{align*}
\int_{\mathbb{R}^{2d}} |V_\varphi g (z)|^2|z|  \ud z.
\end{align*}
Thus, in practice, the choice of the reconstruction window $\varphi$ amounts to a rough guess of the time-frequency profile of the model window $g$,  {with $\varphi=g$ being the ideal choice. Without knowledge of $g$,
	a natural choice for the reconstruction window is the standard Gaussian $\varphi(t)=2^{1/4}e^{-\pi t^2}$, or its time-frequency shifted version $\varphi(t)=2^{1/4}e^{-\pi (t-x)^2}e^{-2\pi i \xi t}$ if an educated guess for the time-frequency center $(x,\xi)$ of $g$ is available. See Figure~\ref{fig:1} for an illustration with $\varphi=g$ or $\varphi\neq g$.}

As an application of Theorem \ref{mthm1} we see that the measure of the estimation error is dominated by the perimeter of the masking domain $\abs{\partial \Omega}$.
\begin{theorem}\label{mcoro1}
Under the conditions of Theorem \ref{mthm1}, with the stated probability, the estimation error satisfies
\begin{equation*}
	\big|{\Omega \triangle \widehat{\Omega}}\big| \leq C_1 \cdot |\partial\Omega|,
\end{equation*}
where the constant $C_1$ depends on $\varphi$ and $g$ but not on $\Omega$.
\end{theorem}

Finally, we also derive a bound for the expected reconstruction error.
\begin{theorem}\label{mthm2}
Under the conditions of Theorem \ref{mthm1}, there exist constants $C_1, C_2, r>0$ that depend on $g$ and $\varphi$ but not on $\Omega$ such that
if
\begin{align}\label{eq_o1}
K \geq C_1 \log\big(\abs{\Omega^r} \big),
\end{align}
then
\begin{equation*}
	\mathbb{E} \left\{ \big|{\Omega \triangle \widehat{\Omega}}\big|\right\} \leq C_2 \cdot |\partial\Omega|.
\end{equation*}
\end{theorem}
 {The constant $r$ in Theorem \ref{mthm2} plays a moderate role for regular domains $\Omega$. For example, in the asymptotic dilation regime, where $\Omega$ is obtained by increasing dilation of a given smooth domain, $|\Omega^r|$ can often be replaced by $|\Omega|$.}

 {Our proofs below provide more precise versions of the main results, with additional qualitative information on the various constants. In particular, Theorem~\ref{thm:error-expec} below is a more detailed version of Theorem~\ref{mthm2}, while Remark \ref{rem_cons2} provides additional information on the constants $C_1, C_2, r$.}

\subsection{Real-valued noise}\label{sec_real}
We now briefly discuss how to adapt the main result to \emph{real} white noise. Consider the estimation problem in Section \ref{sec_r1}, but this time under the assumption that the ambient noise is real-valued and white. Specifically, suppose that we observe $H_\Omega(\mathcal{N}_1), \ldots, H_\Omega(\mathcal{N}_K)$, where
$\mathcal{N}_1, \ldots, \mathcal{N}_K$ are independent copies of real-valued white noise on $\mathbb{R}^d$ with variance $\sigma^2>0$.

Setting $K' := \lfloor K/2 \rfloor$, the new stochastic processes
\begin{align}\label{eq_complexification}
\mathcal{N}'_k := \mathcal{N}_k + i \cdot \mathcal{N}_{k+K'}, \qquad k=1, \ldots, K',
\end{align}
are independent copies of complex white noise with variance $2 \sigma^2$. As the averaged spectrogram \eqref{eq_asp} associated with the observations $H_\Omega(\mathcal{N}'_1), \ldots, H_\Omega(\mathcal{N}'_{K'})$ can be computed from the available data,
\begin{align}\label{eq_real_estim}
\rho(z):=\frac{1}{K'}\sum_{k =1}^{K'} \big|V_\varphi \big(H_\Omega \No'_k \big)(z) \big|^2 = \frac{1}{K'}\sum_{k =1}^{K'} \big|V_\varphi\big(H_\Omega \No_k\big)(z) + i \cdot V_\varphi\big(H_\Omega \No_{k+K'}\big)(z)\big|^2,
\end{align}
and $K' \leq K \leq 3 K'$ as soon as $K \geq 4$, the performance guarantees from Theorems~\ref{mthm1}, \ref{mcoro1}, and  \ref{mthm2}  still hold, with possibly larger constants.

Finally, note that if $\mathcal{N}_1, \ldots, \mathcal{N}_K$
are independent copies of complex-valued white noise, then so are
$\mathcal{N'}_1, \ldots, \mathcal{N}'_{K'}$ as given by \eqref{eq_complexification}. Hence, the \emph{complexified estimator} \eqref{eq_real_estim} satisfies the performance guarantees of Theorem \ref{mthm1} regardless of whether the ambient noise is real or complex-valued.

\begin{figure}[ht]
\begin{center}
\includegraphics[width=6cm,height=5.76cm]{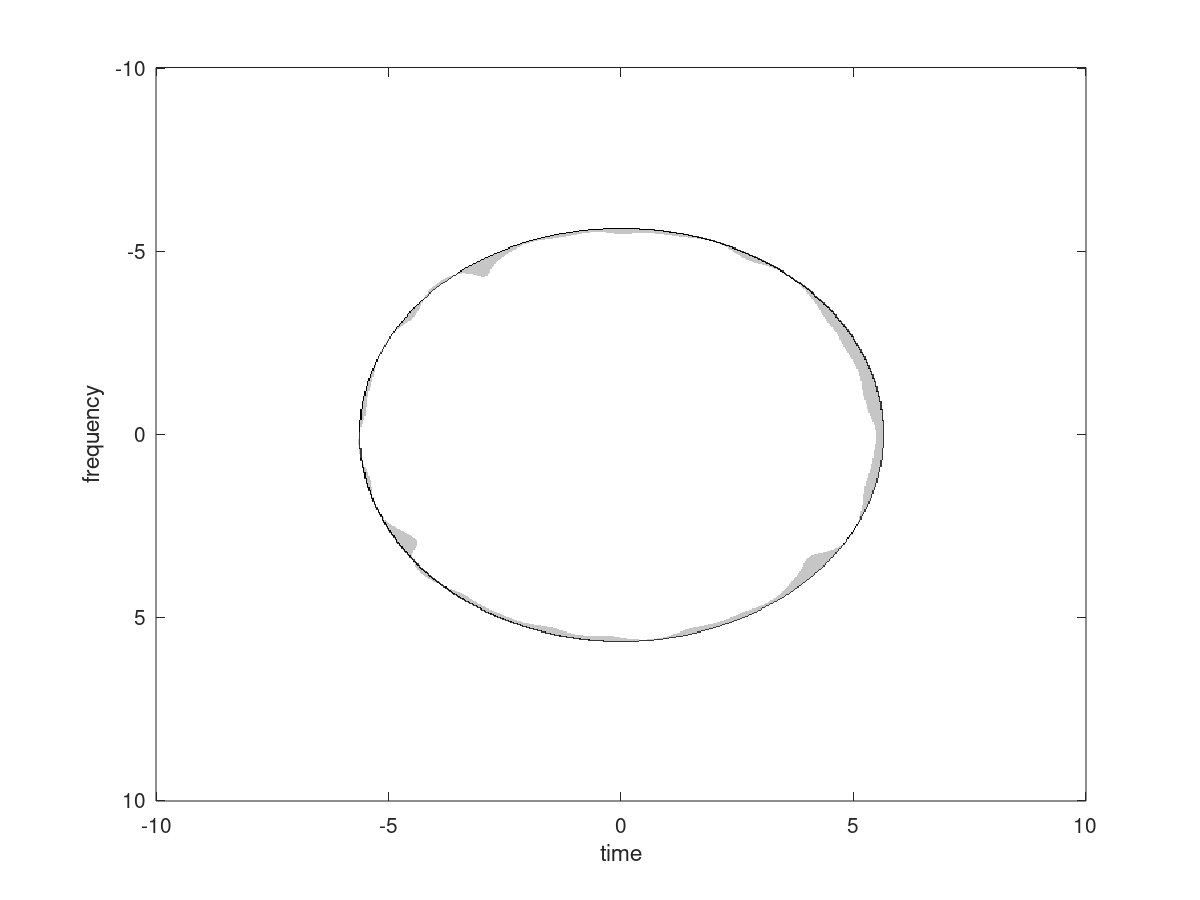}
\includegraphics[width=6cm,height=5.76cm]{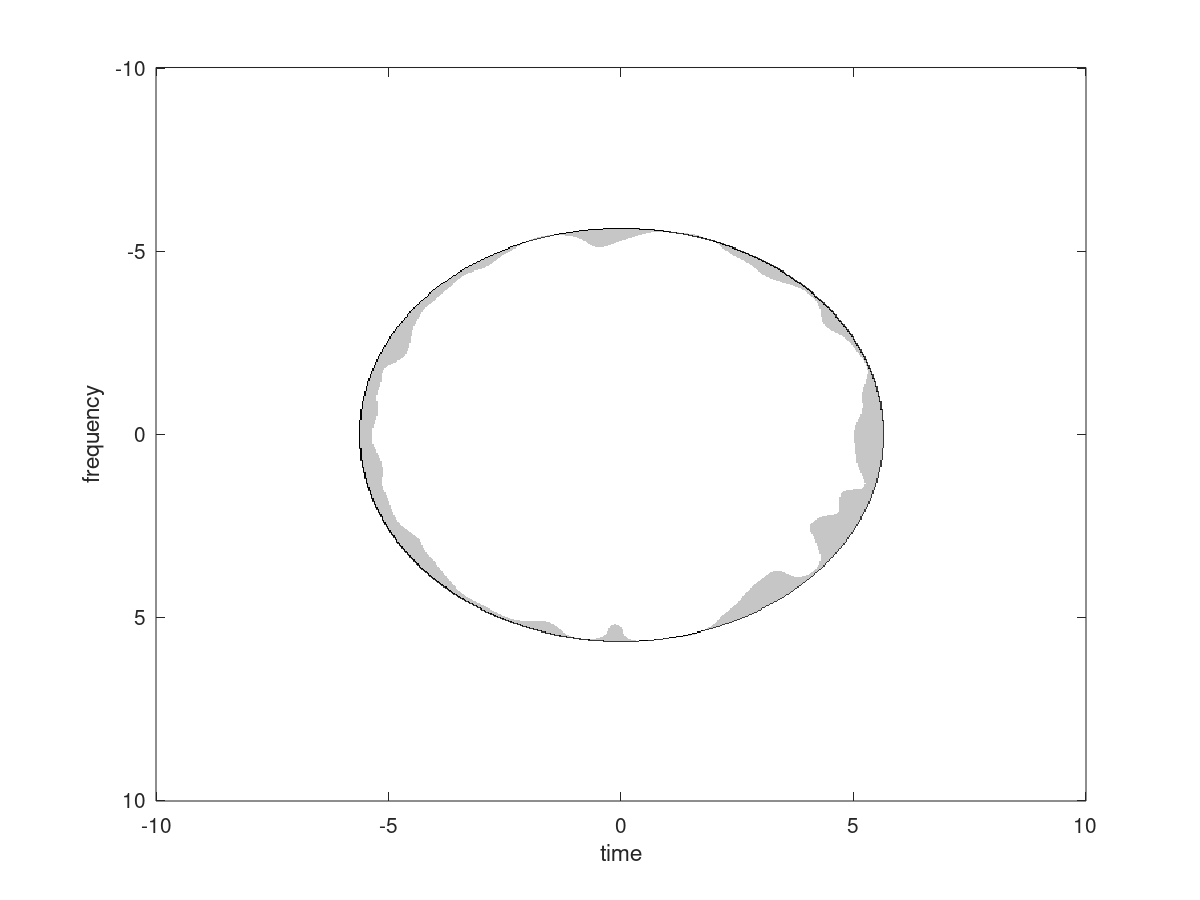}
\\
\includegraphics[width=6cm,height=5.76cm]{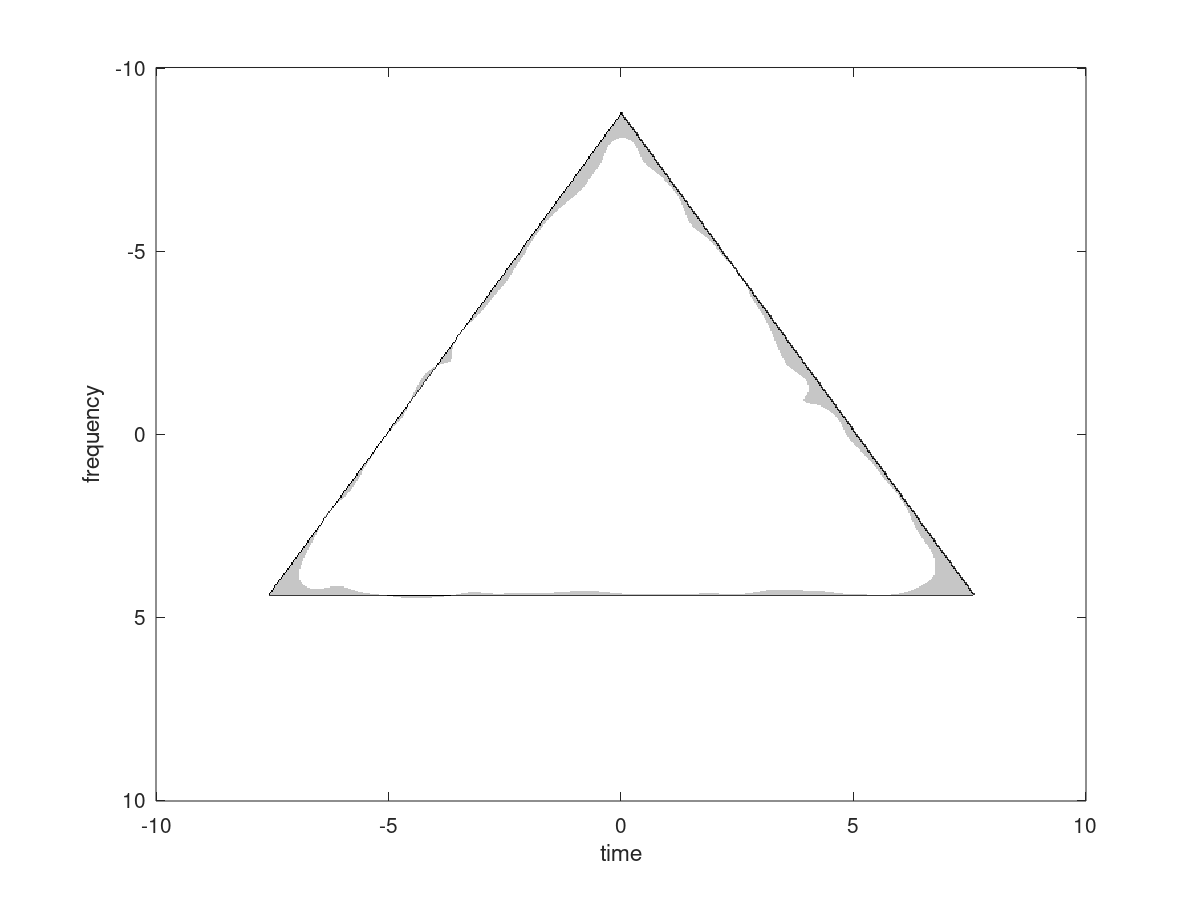}
\includegraphics[width=6cm,height=5.76cm]{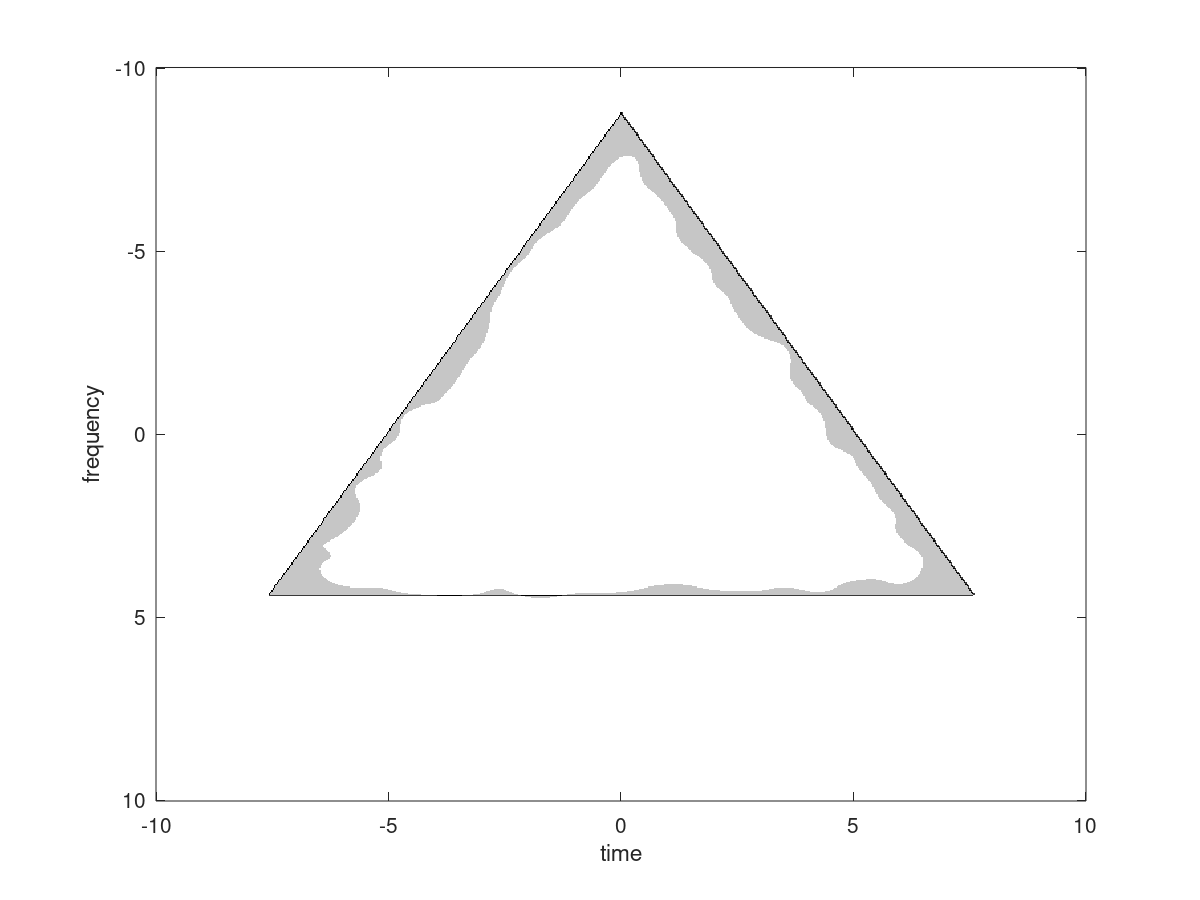}
\caption{The symmetric difference $\Omega\Delta\widehat{\Omega}$ is depicted in grey and $\partial \Omega$ in black.\\
Parameters:  $|\Omega|=100$; $K=20$; $\varphi(t)=2^{1/4}e^{-\pi t^2}$; left column: $g=\varphi$; right column: $g(t)= \varphi(t)t^2$.}
\label{fig:1}
\end{center}
\end{figure}

\subsection{Related literature and technical overview}\label{sec_d}
The spectrogram of a function is the squared absolute value of its short-time Fourier transform \eqref{eq_stft} and provides a useful proxy for the elusive notion of a \emph{time-varying spectrum} \cite{MR1602835}. As with other quadratic time-frequency representations in the so-called \emph{Cohen class},
its use represents a compromise between many desirable but only approximately achievable properties \cite{MR4182488}. A specific merit of the spectrogram is being non-negative and efficiently computable, which makes it a popular tool in signal processing \cite{papandreou2018applications}. On the other hand, the choice of the analyzing window $g$ introduces a trade-off between spatial and temporal resolution. In the acoustic filter model \eqref{eq_Hm}, $g$ is a \emph{partially unknown degree of freedom}. An important consequence of our work is that the estimation of a binary mask succeeds based only on a rough educated guess for $g$ (namely, the reconstruction window $\varphi$).

The estimator that we propose and analyze is a level set of the average observed spectrogram \eqref{eq_asp} taken with an \emph{observation dependent threshold} (lower quantile). Level sets of spectrograms are standard tools in signal processing. For example, a common method to extract a signal embedded into noise consists in computing a spectrogram level set with a carefully chosen threshold related to the expected signal to noise ratio. See \cite{golisu22} for performance guarantees for signal estimation under white noise contamination based on spectrogram level sets with Gaussian windows, and \cite{MR4162314} for a systematic study of analytic time-frequency representations of white noise.

At the technical level, we make the central assumption that the ambient noise is \emph{Gaussian}, which translates into suitable concentration properties for various statistical quantities. Indeed, we make essential use of concentration of measure for quadratic forms \cite{ada15,ruve13}. In order to apply these in an analog setting, we consider adequate low-rank approximations of the filter \eqref{eq_Hm}.

It is a general principle that the low-rank core of a matrix can be revealed by sketching it with independent random vectors. One usually takes a number of random input vectors slightly exceeding the matrix numerical rank and then computes a reduced singular value decomposition \cite{MR2806637}. Remarkably, as we show in Theorem \ref{mthm1}, a time-frequency filter with a \emph{binary} mask can be learned from far fewer noise measurements, and without costly linear algebra computations. (Indeed, the numerical rank of $H_\Omega$ is $\approx \abs{\Omega}$ \cite[Chapter 1]{MR2883827} while the number of required measurements is $\approx \log(\abs{\Omega})$.) The proof of Theorem \ref{mthm1} exploits a priori knowledge of the \emph{asymptotic profile of the spectrum and eigenvectors} of the filter $H_\Omega$, which reduces the size of the learning task. At the technical level, we build and significantly elaborate on the spectral analysis of \cite{abgrro16, abpero17, baconi20}.

Finally, we point out some limitations and possible extensions of our work. First, while noise can be multifaceted \cite{1550188}, ours is white and Gaussian. Second, we work with signals directly in the continuous domain, while in practice measurements are finite, and effective duration and bandwidth need to be taken into account \cite{MR2883827,papandreou2018applications}. These are modeling simplifications and we expect that our work can be suitably generalized.  {For example, our methods can be suitably extended to ambient noise with finite bandwidth $W$, modeled as $\mathcal{N}_W := \mathcal{F}^{-1} (1_{[-W/2,W/2]^d} \mathcal{F} \mathcal{N})$, where $\mathcal{F}$ is the Fourier transform and $\mathcal{N}$ is white noise, provided that the frequency mask $\Omega$ satisfies $\Omega \subset \mathbb{R}^d \times 1_{[-W/2,W/2]^d}$ --- as, in this case, $H_\Omega \mathcal{N}$ is approximately equal to $H_\Omega \mathcal{N}_W$. In this respect, we mention the different but related problem of \emph{designing} suitably smooth and fast decaying input signals for operator identification \cite{MR3207671}.}

\subsection{Organization} Section \ref{sec_pre} establishes the notation and preliminary results. Basic background on time-frequency analysis, noise and spectral theory is provided in Section \ref{sec_tff}. The main results are then proved in Section \ref{sec_main}.

\section{Preliminaries}\label{sec_pre}
\subsection{Notation}
We denote  balls in $\R^d$ by $B_R(z)$ and define the distance between a point $z$ and a set $\Omega$ by $\text{dist}(z,\Omega)=\inf\{|z-w|:\ w\in\Omega\}$. Moreover, we write $A\lesssim B$ if there exists a universal constant $C>0$ such that $A\leq CB$. We use $\|\cdot\|_s$ and $\|\cdot\|_F$ for the spectral and Frobenius norms of a matrix in $\C^{m\times m}$ respectively. The convolution of two functions $f,g$ on $\mathbb{R}^d$ is $f\ast g(z)=\int_{\R^d} f(w)g(z-w)\ud w$.

\subsection{Sets with finite perimeter}\label{sec_fp}
For a compact set $\Omega \subset \mathbb{R}^{2d}$ we let $\abs{\partial \Omega}$ denote the $(2d-1)$-dimensional Hausdorff measure of the topological boundary of $\Omega$. Whenever this measure is finite, the set $\Omega$ has \emph{finite perimeter} in the sense that its characteristic  function $\chi_\Omega$ has bounded variation. Moreover,
in this case, $\mathrm{Var}(\chi_\Omega) \leq \abs{\partial \Omega}$, as $\mathrm{Var}(\chi_\Omega)$ is the $(2d-1)$-dimensional Hausdorff measure of the measure-theoretic boundary of $\Omega$, which is a subset of the topological boundary of $\Omega$. In our main results, the assumption that $\Omega$ is compact, together with \eqref{eq_as2}, 
implies that $\abs{\partial \Omega} < \infty$. See \cite[Chapter 5]{evga92} for comprehensive background on sets of finite perimeter.

We will use the following estimate on regularization with an integrable convolution kernel $\psi:\R^{2d}\rightarrow \R$:
\begin{align}\label{eq_reg}
\big\|\chi_\Omega\ast \psi -\big(\textstyle{\int}\psi\big)\chi_\Omega\big\|_1\leq \int_{\R^{2d}}|\psi(z)| |z|\ud z\cdot|\partial\Omega|,
\end{align}
for a proof, see, e.g., \cite[Lemma~3.2]{abgrro16}.

\subsection{Complex Gaussian vectors}
A complex Gaussian vector is a random vector $X$ on $\mathbb{C}^n$ such that $(\text{Re} (X), \text{Im} (X))$ is normally distributed, has zero mean, and vanishing pseudo-covariance $\mathbb{E}\big\{ X   X^t \big\} = 0$. The stochastics of such $X$ are thus encoded in its covariance matrix
\begin{align*}
	\Sigma = \mathbb{E} \big\{ X   X^* \big\},
\end{align*}
and we write $X\sim \mathcal{CN}(0,\Sigma)$.

\subsection{Complex white noise}\label{sec_noise}
Let $\mathcal{S}_\R(\R^d)$ and $\mathcal{S}(\R^d)$ be the spaces of real and complex valued Schwartz functions and $\S'_\R(\R^{d})$, $\S'(\R^d)$ their respective duals (the spaces of tempered distributions). 
By the Bochner-Minlos theorem, there exists a unique probability measure on $\mathcal{B}_\R$, the family of Borel subsets of $\S'_\R(\R^d)$, such that 
$$
\int_{S_\R'(\R^d)} e^{i\langle f,\varphi\rangle}\ud\mu_\R(\varphi)=e^{-\frac{\sigma^2}{4}\|f\|^2},
$$
for all $f\in\S_\R(\R^d)$. 
With the identification $\S(\R^d)=\S_\R(\R^d)+i\S_\R(\R^d)\simeq\S_\R(\R^d)\times\S_\R(\R^d)$, we introduce the product measure $\mu_\C=\mu_\R\times \mu_\R$, defined on $\mathcal{B}_\C$, the product $\sigma$-algebra on $\S^\prime(\R^d)$. The measure space $\big(\S^\prime(\R^d),\mathcal{B}_\C,\mu_\C\big)$ is called \emph{complex white noise}. Intuitively, we think of this space as the law of a random distribution $\mathcal{N}$.
If $f,g\in \S(\R^d)$ and $\mathcal{N}$ is an instance of complex white noise, then  $\langle f,\mathcal{N}\rangle \sim \mathcal{C}\No\big(0,\sigma^2\|f\|^2\big)$
and 
$\E\left\{\langle f,\mathcal{N}\rangle \overline{\langle g,\mathcal{N}\rangle}\right\}=\sigma^2\langle f,g\rangle.$
For more on complex white noise, see for example \cite[Section 6]{hida}.

\subsection{Concentration of measure for quadratic forms}
We will need the following complex version of of the Hanson-Wright inequality \cite{ada15,ruve13}, which we prove for completeness.
\begin{theorem}\label{thm:hansonwright}
	Let $X$ be an $n$-dimensional complex Gaussian random variable with $X\sim \mathcal{CN}(0,\Sigma)$ and $A\in\C^{n\times n}$ Hermitian. Then for every $t>0$,
	$$
	\P\big(|\langle AX,X\rangle -\E\{\langle AX,X\rangle\} |\geq t\big)\leq 2\ \emph{exp}\left(-C_{\mathrm{hw}}\min\left(\frac{t^2}{\|\Sigma\|^2_s\|A\|_F^2},\frac{t}{\|\Sigma\|_s\|A\|_s}\right)\right)
	$$
	for some universal constant $C_{\mathrm{hw}}>0$.
\end{theorem}
\begin{proof}
By absorbing the matrix $\Sigma$ into $A$ we can assume without loss of generality that $\Sigma$ is the identity matrix.
Define   $Y:=(\text{Re}(X),\text{Im}(X))^T\in \R^{2n}$ as well as
$$
\mathbb{A}:=\begin{pmatrix}
	\text{Re}(A)&-\text{Im}(A)\\ \text{Im}(A)& \text{Re}(A)
\end{pmatrix}\in\R^{2n\times 2n}.
$$
As $A$ is Hermitian, one has that $\mathbb{A}$ is symmetric and $\langle AX,X\rangle =\langle \mathbb{A}Y,Y\rangle$ which implies that $\|A\|_s=\|\mathbb{A}\|_s$. On the other hand, it is clear from its definition that $\|\mathbb{A}\|_F^2=2\|A\|_F^2$.  Moreover, if $X\sim C\mathcal{N}(0,I)$ then $Y\sim \mathcal{N}\big(0,\tfrac{1}{2}I\big)$. The result then follows after an application of \cite[Theorem 2.5]{ada15}.
\end{proof}

\section{Time-frequency filters and the range of the short-time Fourier transform}\label{sec_tff}

\subsection{Basic properties of the STFT}

Let $z=(x,\xi)\in\R^d\times\R^d$. The \emph{time-frequency shifts} of a function $f\in L^2(\R^d)$ are \[\pi(z)f(t)=e^{2\pi i \xi\cdot t}f(t-x), \qquad z=(x,\xi) \in \mathbb{R}^d\times\mathbb{R}^d.\]
The short-time Fourier transform (STFT) of a function $f \in L^2(\mathbb{R}^d)$ with respect to a window $g \in L^2(\mathbb{R}^d)$ is defined by \eqref{eq_stft} and satisfies the following \emph{orthogonality relation}
\begin{equation}
\label{eq:ortho-rel}
\int_{\R^{2d}} V_{g_1}f_1(z)\overline{V_{g_2}f_2(z)}\ud z=\langle f_1,f_2\rangle \langle g_2,g_1\rangle,\quad f_1,f_2,g_1,g_2\in L^2(\R^d).
\end{equation}
The range of the STFT, $V_g(L^2(\R^d))$, is a reproducing kernel subspace of $L^2(\R^{2d})$ with \emph{reproducing kernel} 
$K_g(z,w)=\|g\|^{-2}_2\langle \pi(w)g,\pi(z)g\rangle$. 
Concretely, this means that $V_g(L^2(\R^d))$ is closed in $L^2(\R^{2d})$ and the following reproducing formula holds:
\begin{equation}\label{eq:rep} 
V_g f(z)=\int_{\R^{2d}}V_gf(w)K_g(z,w)\ud w, \qquad f \in L^2(\mathbb{R}^d).
\end{equation}
As soon as $g \in \mathcal{S}(\mathbb{R}^d)$, the definition of the short-time Fourier transform \eqref{eq_stft} extends to $f \in \mathcal{S}'(\mathbb{R}^d)$ in the usual way and yields a continuous function $V_g f: \mathbb{R}^{2d} \to \mathbb{C}$; see \cite[Chapter 11]{groe1}, or \cite{benyimodulation}. 

\subsection{Spectral properties of time-frequency filters}\label{sec_tf}
Let $g\in\S(\R^d)$ with $\norm{2}{g}=1$ and $\Omega\subset\R^{2d}$ be compact. The \emph{time-frequency localization operator} $H_\Omega:L^2(\mathbb{R}^d) \to L^2(\mathbb{R}^d)$ is defined by the formula:
\begin{align}\label{eq_locop}
H_\Omega f= \int_\Omega V_gf(z)\pi (z)g \ud z,
\end{align}
where
the integral is to be interpreted in the weak sense
\begin{align*}
	\langle H_\Omega f, h \rangle = \int_\Omega V_gf(z) \langle \pi (z)g,h\rangle \ud z, \qquad h \in L^2(\mathbb{R}^d).
\end{align*}
Originally introduced in \cite{da88}, $H_\Omega$ is a basic tool in signal processing that models a non-stationary filter. As an operator, $H_\Omega$ is compact and self-adjoint. {In terms of the isometry $V_g:L^2(\mathbb{R}^d) \to L^2(\mathbb{R}^{2d})$, the time-frequency localization operator can be factored as $H_\Omega f= V_g^*(\chi_\Omega \cdot V_g f)$, which shows that $V_g H_\Omega f$ is the orthogonal projection of $\chi_\Omega V_g f$ onto $V_g (L^2(\mathbb{R}^d))$. (In fact, the same observation applies to any mask $m$ in lieu of $\chi_\Omega$ and shows the equivalence of \eqref{eq_Hm_extremal} and \eqref{eq_Hm}.)}

Let $1\geq \lambda_1\geq \lambda_2\geq ....>0$ be its non-zero eigenvalues counted with multiplicities and $\{f_n\}_{ n \in \N} \subset L^2(\mathbb{R}^d)$ the corresponding orthonormal set of eigenfunctions. This set may be incomplete if $H_\Omega$ has a non-trivial null-space --- it could also be finite if $H_\Omega$ has finite rank, but we still write $n \in \mathbb{N}$ with a slight abuse of notation. In any case, $H_\Omega$ has infinite rank as soon as $\Omega$ has non-empty interior \cite[Lemma 5.8]{doro14}. The  set of eigenfunctions are not only orthogonal on $\R^{d}$ but they also satisfy the following \emph{double orthogonality} property:
\begin{equation}
\label{eq:double-orth}
\int_\Omega V_g f_n(z)\overline{V_g f_{k}(z)}\ud z =\lambda_n\delta_{n,k}.
\end{equation}
It is easy to see that $H_\Omega$ is trace-class and 
\begin{align}\label{eq_trace}
\sum_{n\geq 1} \lambda_n = \abs{\Omega}.
\end{align}
The investigation of the fine asymptotics of the spectrum of $H_\Omega$ is an active field of research.
 Roughly stated, $H_\Omega$ has $\approx \abs{\Omega}$ eigenvalues $\lambda_n$ that are close to 1, while $\lambda_n$ decays fast as soon as $n$ is moderately larger than $\abs{\Omega}$ \cite{da88}. We will only need the following simple estimate, whose proof follows from a minor variation of \cite[Lemma~4.3]{maro21} and is omitted.
\begin{lemma}\label{lem:aux-eigenvalues}
	Let $g\in\S(\R)$ satisfy $\norm{2}{g}=1$, and let $\Omega\subset\R^{2d}$ be a compact set satisfying \eqref{eq_as2}.
	Then $\lambda_n\geq \frac{3}{4}$ for every $1\leq n\leq  |\Omega|/2$.
\end{lemma}

The definition of $H_\Omega$ readily extends to distributional arguments: if $f \in \mathcal{S}'(\mathbb{R}^d)$, then $V_g f$ is a continuous function, and \eqref{eq_locop} is well-defined.
It follows from basic pseudo-differential theory that
\begin{align}\label{eq_1}
H_\Omega: \mathcal{S}'(\mathbb{R}^{d}) \to \mathcal{S}(\mathbb{R}^{d}),
\end{align}
and
\begin{align}\label{eq_dual}
\langle H_\Omega f, h \rangle = \langle f, H_\Omega h \rangle, \qquad f,h \in \mathcal{S}'(\mathbb{R}^{d}).
\end{align}
(Indeed, the \emph{Weyl symbol} of $H_\Omega$ is the Schwartz function $\chi_\Omega * W_g$, where $W_g$ is the Wigner distribution of $f$ ---  see \cite[Chapter 2 and Theorem 2.21]{folland}.)
In particular, $f_n \in \mathcal{S}(\mathbb{R}^d)$, for all $n \geq 1$; see \cite{ baconi20,cogroe03,MR4201879} for more refined regularity properties.

\subsection{Filtering white noise}
Almost every realization of white noise $\mathcal{N}$ defines a tempered distribution, and, by \eqref{eq_1}, $H_\Omega \mathcal{N}$ is a well-defined random Schwartz function on $\mathbb{R}^{2d}$ --- see \cite[Section 3]{MR4162314}, \cite[Section 6.1]{gwhf} and
\cite[Section 3]{golisu22} for more on the random function $V_g \mathcal{N}$.

We now note that the eigenexpansion of $H_\Omega$ can be applied to white noise.

\begin{lemma}\label{lemma_sw}
	Let $g\in\S(\R^d)$ with $\norm{2}{g}=1$, $\Omega\subset\R^{2d}$ compact, and $\mathcal{N}$ complex white noise with standard deviation $\sigma$. Then there exists a sequence $\alpha_n\sim \mathcal{C}\No(0,\sigma^2)$ of independent complex normal variables such that, almost surely,
	\begin{align}\label{eq_series}
		H_\Omega (\mathcal{N}) = \sum_{n=1}^\infty \lambda_n \alpha_n f_n,
	\end{align}
	with (almost sure) absolute convergence in $L^2(\mathbb{R}^{2d})$.
\end{lemma}
\begin{proof}
Let 
$\alpha_n := \langle \mathcal{N}, f_n \rangle$. The orthogonality of $\{f_n: n \geq 1 \}$ implies the independence of $\{\alpha_n: n \geq 1\}$. 
By \eqref{eq_trace}, $\sum_n \lambda_n \mathbb{E}\{|\alpha_n| \} < \infty$. Hence, almost surely, $\sum_n \lambda_n |\alpha_n| < \infty$. Consider one such realization. Then $\sum_n \lambda_n \alpha_n f_n$ is absolutely convergent in $L^2(\R^d)$. In addition, for $k \geq 1$,
\begin{align*}
\Big\langle \sum_n \lambda_n \alpha_n f_n , f_{k}\Big\rangle = \lambda_{k} \alpha_{k},
\end{align*}
by absolute convergence, while 
\begin{align*}
\langle H_\Omega \mathcal{N}, f_{k} \rangle = 
\langle \mathcal{N}, H_\Omega f_{k} \rangle
= \lambda_{k} \alpha_{k}.
\end{align*}
On the other hand, $\langle \sum_n \lambda_n \alpha_n f_n , h\rangle=0$ for $h \in \ker(H_\Omega) =\text{span} \{f_n: n \geq 1\}^\perp$, while \eqref{eq_dual} gives
\begin{align*}
\langle H_\Omega \mathcal{N}, h \rangle = 
\langle \mathcal{N}, H_\Omega h \rangle= 0.
\end{align*}
This proves \eqref{eq_series}.
\end{proof}

\section{Proof of the main results}\label{sec_main}
Throughout this section we let $g,\varphi \in \mathcal{S}(\mathbb{R}^d)$ be normalized by $\norm{2}{g}=\norm{2}{\varphi}=1$. We also let $\Omega \subseteq \mathbb{R}^{2d}$ be compact with $\abs{\partial{\Omega}} < \infty$ and satisfying \eqref{eq_as2}. The corresponding time-frequency localization operator is diagonalized as in Section \ref{sec_tf}; we adopt the corresponding notation. We also let
$\mathcal{N}_1, \ldots, \mathcal{N}_K$ be $K$ independent realizations of complex white noise with variance $\sigma^2>0$.

\subsection{Proof strategy}
 {
	Let us outline the overall strategy to prove the main results (Theorems~\ref{mthm1}--\ref{mthm2}).\\
	\noindent\textbullet\ We initially assume that $\sigma=1$ and study the level sets
	\begin{align}\label{eq_level}S_\delta=\{z\in\R^{2d}:\ \rho(z)\geq \delta\}\end{align}
	associated with a \emph{deterministic threshold} $\delta \in (0,1)$. The core of the proof consists in constructing an exceptional set $\Gamma=\Gamma_\delta \subset\R^{2d}$ that satisfies the following with high probability:
	\begin{enumerate}
		\item[(i)] $\chi_\Omega(z)=\chi_{S_\delta}(z)$, for all $z\in\R^{2d}\backslash\Gamma$, (Theorems~\ref{thm:error-outside} and \ref{thm:error-inside});
		\item[(ii)] $\Gamma\subset [\partial\Omega]^r$, for an adequate constant $r>0$ independent of $\Omega$, (Propositions~\ref{prop:estimates} and \ref{prop:size-Ag^c}).
\end{enumerate}}

 {Part (i) is in turn proved in two steps. We first estimate the failure probabilities $\mathbb{P}(\chi_\Omega(z) \not= \chi_{S_\delta}(z))$ related to each individual $z\in\R^{2d}\backslash \Gamma$ (Lemma~\ref{lem:aux-hw} and Theorem~\ref{thm:high-prob}). Second we exploit the moduli of continuity of the short-time Fourier transform to strengthen the previously obtained estimates and derive uniform conclusions. At the technical level, this is achieved by studying an adequate local maximum operator applied to the reproducing kernel of the range of the short-time Fourier transform (Lemmas \ref{lem:aux-local-av} and \ref{lem:sup}), which allows us to implement a net argument. See also the discussion in Section \ref{sec_m}.}

 {\noindent\textbullet\ To connect the previous analysis to the quantile set \eqref{eq_estimator}, we show that  $\|\rho\|_\infty\asymp 1$ with high probability (Proposition~\ref{prop:max-rho}, Sections~\ref{sec_5.3}--\ref{sec_5.5}).}

 {\noindent\textbullet\ Finally, since the definition of $\widehat{\Omega}$ is homogeneous w.r.t. $\sigma$, all results extend to general noise levels.}

To formulate our arguments, we introduce the auxiliary function $\vartheta:\R^{2d}\rightarrow \R_{\geq 0}$:
\begin{equation}\label{eq:def-rho-the}
	\qquad\vartheta(z):=\sum_{n\in\N}|\lambda_n V_\varphi f_n(z)|^2,
	\qquad z \in \mathbb{R}^{2d}.
\end{equation}
As $\{V_gf_n\}_{n\in\N}$ is an orthonormal sequence in the reproducing kernel Hilbert space $V_\varphi(L^2(\R))$,
\begin{equation}\label{eq:sigma-bound}
	\vartheta(z)\leq \sum_{n=1}^\infty |V_\varphi f_n(z)|^2 \leq K_\varphi(z,z)=1,
\end{equation} 
while \eqref{eq_trace} shows that $\vartheta$ is not identically zero.

\subsection{Pointwise error estimate}
We start by investigating the stochastic concentration of the functions $\rho$ and $\vartheta$.
 
 \begin{lemma}\label{lem:aux-hw}
Let $\sigma=1$. There exists a universal constant  $C_1\in (0,1]$ such that for every $z\in\R^{2d}$:
\begin{align}\label{eq:hp1}
\P\big(|\rho (z)-\vartheta (z)|\geq t\big) 
&\leq 3\ \emph{exp}\left(-C_1  K\min\left(\frac{t^2}{ \vartheta(z)^2},\frac{t}{ \vartheta(z)}\right)\right).
\end{align}
\end{lemma} 
\noindent \proof
  Let $L\in \N$ and define the complex Gaussian random vector
  $
  X_n:= \alpha_{\text{mod}(n-1, L)+1}^{\lceil  n /L\rceil},$ for $n\in\{1,...,L K\}$. Clearly, $X\sim \mathcal{CN}\big(0,  I_{K  L}\big)$. Let us also define the matrix-valued function $M:\R^{2d}\rightarrow \R^{ L\times  L}$,
$$
{M}(z)_{\ell,n}:=\lambda_\ell\lambda_n V_\varphi f_{\ell}(z)\overline{V_\varphi f_{n}(z)},
$$
and the block-diagonal Hermitian matrix $\mathcal{M}:\R^{2d}\rightarrow\R^{KL\times KL}$,
$$
\mathcal{M}(z):=\frac{1}{K}\begin{pmatrix}
 {M}(z) & &0 \\ &\ddots & \\ 0& &  {M}(z) 
\end{pmatrix}\in\C^{ KL\times  KL}.
$$ 
Let $P_L$ denote the projection onto the space spanned by the first $L$ eigenfunctions of $H_\Omega$.
It can then be verified that  
\begin{align*}
\rho_L(z):=\frac{1}{K}\sum_{k=1}^K|V_\varphi (P_L H_\Omega \No_k)(z)|^2&=\frac{1}{K}\sum_{k=1}^K\sum_{\ell=1}^{L} \sum_{n=1}^{L} \alpha_\ell^k\overline{\alpha_n^k} \lambda_{\ell}\lambda_nV_\varphi f_{\ell}(z)\overline{V_\varphi f_n(z)}
\\
&=\langle \mathcal{M}(z) X,X\rangle.
\end{align*}
In order to apply Theorem~\ref{thm:hansonwright}, we need to calculate the expectation of $\langle \mathcal{M}(z) X,X\rangle$ as well as the Frobenius and spectral norm of $\mathcal{M}(z)$.

The independence of the random variables $\alpha_n^k$ implies
\begin{align*}
\mathbb{E}\left\{\rho_L(z)\right\}&=\mathbb{E}\left\{\frac{1}{K}\sum_{k=1}^K|V_\varphi (P_L H_\Omega \mathcal{N}_k)(z)|^2\right\}
\\
&=\frac{1}{K}\sum_{k=1}^K\sum_{\ell=1}^L\sum_{n=1}^L \lambda_\ell\lambda_n \mathbb{E}\left\{\alpha_\ell^k\overline{\alpha_n^k}\right\} V_\varphi f_\ell(z)\overline{V_\varphi f_n(z)}
\\
&=\sum_{n=1}^{L} |\lambda_n V_\varphi f_n(z)|^2=: \vartheta_L(z).
\end{align*}
The Frobenius norm of $\mathcal{M}(z)$ is
\begin{align*}
\|\mathcal{M}(z)\|^2_F&= \frac{1}{K} \| {M}(z)\|^2_F 
=\frac{1}{K}\sum_{\ell=1}^{L}\sum_{n=1}^{L}|\lambda_\ell \lambda_n V_\varphi f_\ell(z) {V_\varphi f_n(z)}|^2
\\
&=\frac{1}{K}\left(\sum_{n=1}^{L}|\lambda_nV_\varphi f_n(z)|^2\right)^2\leq \frac{1}{K}\left(\sum_{n=1}^{\infty}|\lambda_nV_\varphi f_n(z)|^2\right)^2=\frac{1}{K}\vartheta(z)^2.
\end{align*}
The spectral norm of $\mathcal{M}(z)$ is
$
\|\mathcal{M}(z)\|_s =\frac{1}{K}\| {M}(z)\|_s.
$
Thus, by Cauchy-Schwarz,
\begin{align*}
 \| {M}(z)\|_s&=\sup_{\|X\|_2=1}\left(\sum_{n=1}^L\left|\sum_{\ell=1}^L  (M(z))_{n,\ell}X_\ell\right|^2\right)^{\frac{1}{2}}
 \\
 &\leq \left(\sum_{n=1}^L \sum_{\ell=1}^L  |\lambda_\ell\lambda_n V_\varphi f_\ell(z)V_\varphi f_n(z)|^2\right)^{\frac{1}{2}}=\vartheta_L(z)\leq \vartheta(z).
\end{align*}
Therefore, by Theorem~\ref{thm:hansonwright}, 
\begin{align}
\P(\left|\rho_L(z)-\vartheta_L(z)\right|\geq t)
&\leq 2\ \text{exp}\left(-C_{hw} K \min\left(\frac{ t^2}{ \vartheta(z)^2},\frac{  t}{\vartheta(z)}\right)\right). \label{eq:finite-pointwise}
\end{align}
 Let us use the following abbreviations
\begin{align*}
R_1(z)&:=\frac{1}{K}\sum_{k=1}^K \sum_{\ell=L+1}^\infty\sum_{n=1}^L \lambda_\ell \lambda_n \alpha_\ell^k \overline{\alpha_n^k}V_\varphi f_\ell (z)\overline{ V_\varphi f_n(z)},
\\
R_2(z)&:=\frac{1}{K}\sum_{k=1}^K \sum_{\ell=1}^L\sum_{n=L+1}^\infty \lambda_\ell \lambda_n \alpha_\ell^k \overline{\alpha_n^k}V_\varphi f_\ell (z) \overline{V_\varphi f_n(z)},
\\
R_3(z)&:=\frac{1}{K}\sum_{k=1}^K \sum_{\ell=L+1}^\infty\sum_{n=L+1}^\infty \lambda_\ell \lambda_n \alpha_\ell^k \overline{\alpha_n^k}V_\varphi f_\ell (z) \overline{V_\varphi f_n(z)}.
\end{align*}
We now carry out certain estimates for $L \to \infty$, under the assumption that $\lambda_n \not =0$ for all $n \geq 1$. We remark that the conclusions will also hold trivially if $\{\lambda_n\}_{n\in\N}$ is finitely supported, as then $R_1=R_2=R_3=0$ for $L$ large enough.
We first estimate 
  \begin{align*}
\P\left( \left|\rho(z) -\vartheta(z)\right| \geq t  \right)
&\leq \P\left( \left|\rho_L(z)+R_1(z)+R_2(z)+R_3(z)-\vartheta(z)\right| \geq t  \right)
\\
&\leq \P\left( \left|\rho_L(z)- \vartheta(z)\right|+|R_1(z)|+|R_2(z)|+|R_3(z)| \geq t \right)
\\
&\leq \P\left( \left|\rho_L(z)- \vartheta(z)\right|\geq \frac{t}{4} \right)+\sum_{s=1}^3\P\left(|R_s(z)|\geq \frac{t}{4} \right),
\end{align*}
where we used that at least one of the four summands needs to be larger than one fourth of $t/4$ and subsequently the union bound.  

Let us treat the summands separately. First, we  choose $L$ large enough such that $\vartheta_L$ is a good approximation of $\vartheta$. In particular, we assume $\left|\vartheta_L(z)- \vartheta(z)\right|\leq t/8$.  We note here that $L$ may be chosen independently of $z$ since $|V_\varphi f_n(z)|\leq 1$ for every $z\in\R^{2d}$. Then, using  \eqref{eq:finite-pointwise}
we obtain
\begin{align*}
\P\left( \left|\rho_L(z)- \vartheta(z)\right|\geq \frac{t}{4} \right)&\leq \P\left( \left|\rho_L(z)-\vartheta_L(z)\right|\geq \frac{t}{8} \right)
\\
&\leq 2 \cdot\text{exp}\left(-\frac{C_{hw}}{64}K\min\left(\frac{t^2}{\vartheta(z)^2},\frac{t}{\vartheta(z)}\right)\right).
\end{align*}
Let us now estimate the term $R_1$. Then $R_2$ and $R_3$  can be treated with the same arguments. First note that, by \eqref{eq_trace},
\begin{align*}
|R_1(z)|&\leq \frac{1}{K}\sum_{k=1}^K\sum_{\ell=1}^\infty\sum_{n=L+1}^\infty \lambda_n\lambda_\ell |\alpha_n^k\alpha_\ell^k|
\\
&\leq  \frac{1}{K}  \sum_{k=1}^K\left(\sum_{n=L+1}^\infty \lambda_n\right)^{\frac{1}{2}}\left( \sum_{n=L+1}^\infty \lambda_n|\alpha_n^k|^2\right)^{\frac{1}{2}}\left(\sum_{\ell=1}^\infty \lambda_\ell\right)^{\frac{1}{2}}\left(\sum_{\ell=1}^\infty\lambda_\ell |\alpha_\ell^k|^2\right)^{\frac{1}{2}}
\\
&\leq \frac{1}{K} \left(|\Omega|\sum_{n=L+1}^\infty \lambda_n\right)^{\frac{1}{2}} \sum_{k=1}^K  \sum_{\ell=1}^\infty \lambda_\ell|\alpha_{\ell }^k|^2   .
\end{align*}
This estimate implies that 
\begin{align*}
\P\left(|R_1(z)|\geq \frac{t}{4}\right)&\leq \P \left(  \sum_{k=1}^K \sum_{\ell=1}^\infty\lambda_\ell |\alpha_\ell^k|^2\geq \frac{Kt}{4} \left( |\Omega|\sum_{n=L+1}^\infty  \lambda_n\right)^{-\frac{1}{2}}  \right)
\\
&\leq \sum_{k=1}^K \P \left(  \sum_{\ell=1}^\infty\lambda_\ell |\alpha_\ell^k|^2  \geq  \frac{t}{4} \left( |\Omega| \sum_{n=L+1}^\infty  \lambda_n \right)^{-\frac{1}{2}}   \right)
\\
& =  K  \P\left(  \sum_{\ell=1}^\infty \lambda_\ell|\alpha_{\ell }^1|^2 \geq \frac{t}{4}\left(|\Omega|\sum_{n=L+1}^\infty \lambda_n\right)^{-\frac{1}{2}}  \right).
\end{align*}
The probability on the right hand side goes to zero as $L\rightarrow\infty$ (note that the left hand side of the tail bound is independent of $L$).  In particular, we may choose $L$ large enough such that 
\[
\P\left(|R_1(z)|\geq \frac{t}{4}\right)\leq \frac{1}{3}\cdot \text{exp}\left(-\frac{C_{hw}}{64}K \min\left(\frac{t^2}{  \vartheta(z)^2},\frac{t}{ \vartheta(z)}\right)\right).
\]
Repeating the argument for $R_2$ and $R_3$ and choosing $C_1=\min(C_{hw}/64,1)$ concludes the proof.
\pbox

We now turn our attention to the level set \eqref{eq_level}.

\begin{theorem}\label{thm:high-prob}
	
Let $\delta\in(0,\frac{2}{3})$, $\sigma=1$, and $z\in\R^{2d}$ be such that $|\vartheta(z)-\chi_\Omega(z)|\leq \frac{\delta}{4}$. There exists a universal constant $C>0$ such that
\begin{equation}\label{eq:hp2}
\P\big(\chi_\Omega(z)\neq \chi_{S_{\delta}}(z)\big)\leq 3\cdot e^{-C K }.
\end{equation}
If in addition $z\in\Omega^c$, then
\begin{equation}\label{eq:hp4}
\P\big(\chi_\Omega(z)\neq \chi_{S_{\delta}}(z)\big)\leq 3\cdot \emph{exp}\left(- C  \frac{K\delta}{\vartheta(z)}  \right).
\end{equation}
\end{theorem}
\noindent \proof
Let us first reinspect the event $A(z):=\{\chi_\Omega(z)\neq \chi_{S_{\delta}}(z)\}$. If $z\notin\Omega$, then the event $A(z)$ is equivalently described by $\delta\leq \rho(z)$, and, in this case,
$$
\delta\leq  {\rho(z)} =\left| {\rho(z)}- \chi_\Omega(z)\right|\leq \left| {\rho(z)} -\vartheta(z)\right|+\left|\chi_\Omega(z)-  \vartheta(z)\right|.
$$
Hence, by the assumption $\left|\chi_\Omega(z)-  \vartheta(z)\right|\leq \delta/4$  it   follows that
$$
A(z)\subset  
 \left\{ \left| {\rho(z)} -\vartheta(z)\right|\geq  \frac{\delta}{4 }\right\},
$$
 and consequently
\begin{align*}
\P\big(\chi_\Omega(z)\neq \chi_{S_{\delta}}(z)\big)\leq \P\left( \left| {\rho(z)} -  \vartheta(z)\right|\geq \frac{\delta}{4}  \right).
\end{align*}
Therefore, \eqref{eq:hp2} and \eqref{eq:hp4} follow from Lemma~\ref{lem:aux-hw} and the observation $\delta/\vartheta(z)\geq 4$.

If on the other hand $z\in \Omega$, then $A(z)$ is characterized by $\rho(z)< \delta$ and we have
$$
1-\delta< 1-{\rho(z)} =\left| {\rho(z)}-\chi_\Omega(z)\right|\leq\left|\rho(z) - \vartheta(z)\right|+\left|\vartheta(z)- \chi_\Omega(z)\right|.
$$
Again, by the assumption $\left|\chi_\Omega(z)-  \vartheta(z)\right|\leq \delta/4$, and $\delta<\frac{2}{3}$ we obtain
$$
A(z)\subset \left\{ \left| {\rho(z)} -\vartheta(z)\right|>  1-\frac{5\delta}{4} \right\}\subset \left\{ \left| {\rho(z)} -\vartheta(z)\right|>  \frac{1}{6}\right\},
$$
 and  may then proceed as before, this time noting that
 $\vartheta(z) \leq 1$.
\pbox

In Theorem~\ref{thm:high-prob} we required that $|\vartheta(z)-\chi_\Omega(z)|\leq \delta/4$. Let us now bound the measure of the set where this is not satisfied.

\begin{prop}\label{prop:estimates}
Set
\begin{equation}\label{eq_c1}
M({g,\varphi}):=\int_{\R^{2d}}  |V_\varphi g(z)|^2 |z|\ud z.
\end{equation} Then 
\begin{enumerate}[label=(\roman*)]
\item  $\|\chi_{\Omega} -\vartheta \|_1\lesssim  M({g,\varphi})\cdot  |\partial\Omega|$,\label{enum:1}
\item  $\left|\left\{z\in\R^{2d}:\ | \vartheta(z)-\chi_\Omega(z)|\geq \delta \right\}\right|\lesssim   M({g,\varphi})\cdot   \frac{ |\partial\Omega|}{\delta},
$\label{enum:2}
\item \label{enum:3} $\|\vartheta\|_1\leq |\Omega|,$
\item \label{enum:4}  $|\chi_{\Omega}(z) -\vartheta (z)|\lesssim    M({g,\varphi})\cdot \emph{dist}(z,\partial\Omega)^{-1}.$
\end{enumerate}
\end{prop}
\noindent \proof

On the one hand, by the eigenexpansion $H_\Omega=\sum_n \lambda_n  \langle\cdot ,f_n\rangle f_n$,
\begin{align*}
\langle H_\Omega \pi(z)\varphi,\pi(z)\varphi\rangle&=\sum_{n\in\N} \lambda_n  |\langle \pi(z)\varphi ,f_n\rangle|^2
=\sum_{n\in\N} \lambda_n |V_\varphi f_n(z)|^2.
\end{align*}
On the other hand, applying  \eqref{eq_locop} directly gives
\begin{align*}
\langle H_\Omega \pi(z)\varphi,\pi(z)\varphi\rangle
&=\int_{\R^{2d}}\chi_\Omega(w) \left|  \langle \pi (z)\varphi,  \pi (w)g\rangle\right|^2  \ud w
\\
&=\int_{\R^{2d}} \chi_\Omega(w) |V_\varphi g(z-w)|^2 \ud w =\chi_\Omega \ast|V_\varphi g|^2(z), 
\end{align*}
which shows
$$
\sum_{n\in\N} \lambda_n |V_\varphi f_n(z)|^2 =\chi_\Omega \ast|V_\varphi g|^2(z), \qquad z \in \mathbb{R}^{2d}.
$$
We hence estimate 
\begin{align*}
|\chi_\Omega(z)-\vartheta(z)|\leq &\left|\chi_\Omega(z)-\chi_\Omega\ast |V_\varphi g|^2 (z)\right|  +\sum_{n\in\N}(\lambda_n-\lambda_n^2)|V_\varphi f_n(z)|^2.
\end{align*}
If $z\in\Omega$, then 
$$
\sum_{n\in\N}(\lambda_n-\lambda_n^2)|V_\varphi f_n(z)|^2\leq \sum_{n\in\N}(1-\lambda_n)|V_\varphi f_n(z)|^2 \leq \chi_\Omega(z)-\chi_\Omega\ast |V_\varphi g|^2 (z),
$$
where we used that $\sum_n|V_\varphi f_n(z)|^2\leq K_\varphi(z,z)=1$.
If on the other hand $z\in\Omega^c$, then 
$$
\sum_{n\in\N}(\lambda_n-\lambda_n^2)|V_\varphi f_n(z)|^2\leq \sum_{n\in\N}\lambda_n|V_\varphi f_n(z)|^2=\chi_\Omega\ast |V_\varphi g|^2 (z) - \chi_\Omega(z),
$$
and therefore 
\begin{equation}\label{eq:pw.diff}
|\chi_\Omega(z)-\vartheta(z)|\leq 2\left|\chi_\Omega(z)-\chi_\Omega\ast |V_\varphi g|^2 (z)\right|.
\end{equation}
Since $\|V_\varphi g\|_2=1$, by \eqref{eq_reg},
$$
\|\chi_\Omega -\vartheta\|_1\leq 2 |\partial \Omega| \int_{\R^{2d}}   |V_\varphi g(z)|^2 |z|\ud z,
$$
that is, \ref{enum:1} holds. Property \ref{enum:2} then follows by Chebyshev's inequality. The third property follows directly
from \eqref{eq_trace}: $\norm{1}{\vartheta}= \sum_n\lambda_n^2\leq\sum_n \lambda_n=|\Omega|$.

Finally, to prove \ref{enum:4} we define $R=\text{dist}(z,\partial\Omega)$. Then $\chi_\Omega(z)=\chi_\Omega(w)$ for every $|z-w|<R$ which implies using \eqref{eq:pw.diff}
\begin{align*}
\big|\chi_\Omega(z)-\vartheta(z)\big|&\leq  2\big|\chi_\Omega(z)-\chi_\Omega\ast |V_\varphi g|^2 (z)\big|
\\
&\leq 2\big|\chi_\Omega(z)-\chi_\Omega\ast \left(|V_\varphi g|^2\chi_{B_R(0)}\right) (z)\big|+2\big|\chi_\Omega \ast \left(|V_\varphi g|^2(1-\chi_{B_R(0)})\right) (z)\big|
\\
& \leq 2\chi_\Omega(z)\left(1-\int_{|z-w|< R}|V_\varphi g(z-w)|^2\ud w\right)+2\int_{|w|\geq R}|V_\varphi g(w)|^2\ud w.
\end{align*}
Hence,
$$
|\chi_\Omega(z)-\vartheta(z)|\leq 4\int_{|w|\geq R} |V_\varphi g(w)|^2\ud y \leq \frac{4}{R}\int_{\R^{2d}} |V_\varphi g(z)|^2 |z|\ud z.
$$
 \pbox 
 
{
 \begin{rem}
 Since $\vartheta\leq 1$, by Proposition~\ref{prop:estimates}~\ref{enum:4} we see that the estimate \eqref{eq:hp4} improves as the distance between $z$ and $\Omega$ increases.
 \end{rem}
}

\subsection{Estimates of symmetric difference with high probability}\label{sec_m}

 {Our subsequent estimates depend on the local supremum functional
$$\mathcal{K} (z,w):=\sup_{y\in B_1(z)}|K_\varphi(y,w)|,$$ 
which is a standard tool in the theory of reproducing kernel Hilbert spaces, and their associated Banach spaces --- see, e.g., \cite{fegr89}. Since $K_\varphi$ is the integral kernel of a projection operator, it enjoys a self-averaging property: pointwise values of $K_\varphi$ can be obtained as averages thereof. The self-averaging property then allows one to improve estimates that hold on average, and deliver uniform conclusions. In our case, we will use this insight to turn pointwise probability estimates into locally uniform ones.}

As a first step, we bound the probability that $\|(\rho- \vartheta)\mathcal{K}(z_0,\cdot)\|_1$ be large.  {We will afterwards relate this tail bound to the probabilities of the supremum (resp. infimum) of $\rho$ being large (resp. small) on adequate domains --- see Lemma~\ref{lem:sup} and Theorem~\ref{thm:error-inside}.}

\begin{lemma}\label{lem:aux-local-av}
Let $z_0\in\R^{2d}$, $\sigma=1$,  and $\gamma \in[\gamma_0,1],\ \gamma_0>0$. Then there exist constants $c,C>0$ only depending on $\varphi$ and $\gamma_0$ such that 
\begin{align}\label{eq:eq-aux-local-av}
 \P\left( \int_{\R^{2d}}|\rho(w)-  \vartheta(w)| \mathcal{K} (z_0,w)\ud w \geq  \gamma   \right) 
 \leq c\cdot e^{- C   K}  \cdot\int_{\R^{2d}}\vartheta(w) \mathcal{K} (z_0,w)\ud w.  
\end{align}
\end{lemma}
\noindent\proof Let us define 
\begin{equation}\label{eq:Cphi}
C_\varphi :=\int_{\R^{2d}}\mathcal{K}(0,w)\ud w<\infty. 
\end{equation}
Let $p \geq 1$. By Chebyshev's and H{\"o}lder's inequalities,
\begin{align*}
 &\P\left(  \int_{\R^{2d}}|\rho(w)- \vartheta(w)| \mathcal{K} (z_0,w) \ud w\geq   \gamma   \right) 
\\
&\hspace{3cm}\leq  \left(\gamma \right)^{- {p}} \E\left\{ \left(\int_{\R^{2d}}|\rho(w)- \vartheta(w)| \mathcal{K} (z_0,w) {\ud w} \right)^{ {p}}\right\}  
\\
&\hspace{3cm}\leq  \left(  \gamma    \right)^{ -{p}} C_\varphi^{p-1} \int_{\R^{2d}}\E\left\{ |\rho(w)- \vartheta(w)|^{ {p}}  \right\}  \mathcal{K} (z_0,w)\ud w 
\\
&\hspace{3cm}\leq  \left(\frac{C_\varphi }{\gamma  } \right)^{ {p}} \int_{\R^{2d}}\int_0^\infty  { {p}}t^{ {p}-1}\P\left(|\rho(w)- \vartheta(w)| \geq t\right)\ud t\ \mathcal{K}(z_0,w)  \ud w,
\end{align*}
where we used that $C_\varphi\geq 
\int |K_\varphi(0,w)| \ud w \geq 
\int |V_\varphi \varphi (w)|^2 \ud w
=1$.
Let us further estimate the integral with respect to $t$. We invoke Lemma~\ref{lem:aux-hw} 
and use that the respective constant is $C_1\leq 1$ to estimate
\begin{align*}
\int_0^\infty  { {p}}t^{ {p}-1}\P\left(|\rho(w)-  \vartheta(w)| \geq t\right)\ud t &
\\
&\hspace{-2.5cm}\lesssim \int_0^\infty  { {p}}t^{ {p}-1}\text{exp}\left(-C_1  K  \min\left(\frac{t^2}{ \vartheta(w)^2},\frac{t}{  \vartheta(w)}\right)\right)\ud t
\\ & \hspace{-2.5cm}= \int_0^{  \vartheta(w)}  { {p}}t^{ {p}-1}\text{exp}\left(- \frac{C_1 K t^2}{ \vartheta(w)^2}\right)\ud t+\int_{  \vartheta(w)}^\infty p t^{ {p}-1}\text{exp}\left(- \frac{C_1  K t}{ \vartheta(w) }\right)\ud t
\\
& \hspace{-2.5cm}= \frac{p}{2}\left(\frac{ \vartheta(w)^2}{C_1  K }\right)^{\frac{p}{2}}\int_0^{C_1   K}t^{\frac{p-2}{2}}e^{-t}\ud t+p\left(\frac{ \vartheta(w)}{C_1 K }\right)^{ {p} }\int_{C_1  K}^\infty t^{p-1}e^{-t}\ud t
\\
&\hspace{-2.5cm}\leq \vartheta(w)  {C_1}^{-p} \Big[\Gamma\left(\frac{p}{2}+1\right)K^{-\frac{p}{2}}+{\Gamma(p+1)}K^{-p}\Big].
\end{align*}
Hence we have thus far shown that for every $1\leq p<\infty$,
\begin{align*}
\P\left(  \int_{\R^{2d}}|\rho(w)- \vartheta(w)| \mathcal{K}(z_0,w) \ud w\geq   \gamma   \right)&\\ &\hspace{-3cm} \lesssim    \left(\frac{C_\varphi }{C_1\gamma  } \right)^{ {p}}  \left[  \frac{\Gamma\left(\frac{p}{2}+1\right)}{ K ^{ \frac{p}{2}}}+\frac{ \Gamma(p+1)  }{K ^p}\right] \int_{\R^{2d}}\vartheta(w) \mathcal{K}(z_0,w)  \ud w.
\end{align*}
  Let us write for short $s=   C_\varphi /C_1\gamma>1.$  If $ K/ s^2\geq 2$, then we choose  $p=2\lfloor  K/ 2s^2\rfloor\geq 1$ and 
\begin{align*}
s^p\left[  \frac{\Gamma\left(\frac{p}{2}+1\right)}{ K ^{ \frac{p}{2}}}+\frac{ \Gamma(p+1)  }{K ^p}\right] 
 &= \left(\frac{s^2}{   K  }\right)^{ \left\lfloor\frac{  K}{2s^2}\right\rfloor}\left\lfloor\frac{ K}{ 2s^2}\right\rfloor  ! +\left(\frac{s}{  K  }\right)^{2\left\lfloor\frac{ K}{2s^2}\right\rfloor}\left(2\left\lfloor\frac{ K}{2s^2}\right\rfloor \right) !
\\
& \leq \left(\frac{1}{   2  }\right)^{ \left\lfloor\frac{  K}{2s^2}\hspace{-1pt}\right\rfloor}\left\lfloor \frac{   K  }{2s^2}\right\rfloor^{- \left\lfloor\frac{  K}{2s^2}\right\rfloor}\hspace{-1pt}\left\lfloor\frac{ K}{ 2s^2}\right\rfloor  ! + \left(2\left\lfloor\frac{  K  }{2s^2}\right\rfloor\right)^{-2\left\lfloor\frac{ K}{2s^2}\hspace{-1pt}\right\rfloor}\left(2\left\lfloor\frac{ K}{2s^2}\right\rfloor \right) !
\\
 & \leq 2\left(\frac{  1}{2}\right)^{ \left\lfloor\frac{ K}{2s^2}\right\rfloor }, \end{align*}
where the final inequality follows from 
$$
\left({2n}\right)^{-2 n} (2n) != \left({2n}\right)^{ -n }   n  !\cdot \left( {2n}\right)^{ -n}(n+1)\cdot\ldots\cdot 2n\leq \left(\frac{1}{2}\right)^{n},\quad n\in \N.
$$
Since $2\left(\frac{  1}{2}\right)^{ \left\lfloor\frac{ K}{2s^2}\right\rfloor } 
\lesssim e^{- c (C_1\gamma / 2C_\varphi)^{2}{K} }\leq e^{- c (C_1\gamma_0/ 2C_\varphi)^{2} {K}  }$
for a universal constant $c>0$,
we conclude that \eqref{eq:eq-aux-local-av} holds. 

Finally, if $ K/ s^2< 2$, then $\gamma>\gamma_0$ implies that $K<  2C_\varphi^2/\gamma_0^2C_1^2$. Hence, in that case, setting $p=1$ yields
 \begin{align*}
\P\left(  \int_{\R^{2d}}|\rho(w)- \vartheta(w)| \mathcal{K} (z_0,w) \ud w\geq   \gamma   \right) \lesssim    \frac{C_\varphi }{C_1\gamma_0  }   \ \int_{\R^{2d}}\vartheta(w) \mathcal{K}(z_0,w)  \ud w,
\end{align*}
and \eqref{eq:eq-aux-local-av} also holds, possibly after taking a different constants $c,C$. \pbox

As an application, we obtain the following.

\begin{lemma}\label{lem:sup}
Let  $\sigma=1$, $\delta\in [\delta_0,1],\ \delta_0>0$, and  $z_0 \in \mathcal{A}(\delta)$, where
\begin{equation}\label{eq:def-An}
\mathcal{A}(\delta):=\left\{z\in\R^{2d}:\ \int_{\R^{2d}}\vartheta(w)\mathcal{K}(z,w){\ud w}\leq  \frac{\delta}{2\|V_\varphi \varphi\|_1} \right\}.
\end{equation}
 There exist   constants $c,C>0$ only depending on $\varphi$ and $\delta_0$ such that 
\begin{align}
 \P\left(\sup_{z\in B_1(z_0)}\rho(z)\geq \delta\right) 
  \leq c\cdot e^{-C K} \cdot\int_{\R^{2d}}\vartheta(z)  \mathcal{K}(z_0,z) {\ud z} .  
\end{align}
\end{lemma}
\noindent \proof

From the reproducing formula \eqref{eq:rep}, almost surely
\begin{align}\label{eq:est-sup}
\nonumber \sup_{z\in B_1(z_0)}\rho(z)&=\sup_{z\in B_R(z_0)}\frac{1}{K}\sum_{k=1}^K \left|\int_{\R^{2d}}V_\varphi H_\Omega\No_k(w)K_\varphi(z,w)\ud w\right|^2
\\ &\leq \|V_\varphi \varphi\|_1 \sup_{z\in B_1(z_0)}\frac{1}{K}\sum_{k=1}^K  \int_{\R^{2d}}|V_\varphi H_\Omega\No_k(w)|^2|K_\varphi(z,w)|\ud w 
\nonumber
\\ &\leq \|V_\varphi\varphi\|_1\int_{\R^{2d}}\rho(w)\mathcal{K}(z_0,w)\ud w.
\end{align} 
Hence,
\begin{align*}
\P\left(\sup_{z\in B_R(z_0)}\rho(z)\geq \delta\right)&\leq \P\left( \int_{\R^{2d}}\rho(w)\mathcal{K} (z_0,w)\ud w\geq \frac{\delta}{\|V_\varphi\varphi\|_1}\right)
\\
&\hspace{-1cm}= \P\left( \int_{\R^{2d}}\hspace{-0.1cm} (\rho(z)- \vartheta(z))\mathcal{K}(z_0,z) \ud z \geq \left(\frac{\delta}{\|V_\varphi\varphi\|_1}-  \int_{\R^{2d}}\vartheta(z)\mathcal{K} (z_0,z) \ud z \right) \right)
\\
&\hspace{-1cm}\leq  \P\left( \int_{\R^{2d}}|\rho(z)-\vartheta(z)| \mathcal{K} (z_0,z){\ud z}  \geq \frac{\delta}{2\|V_\varphi\varphi\|_1} \right).
\end{align*}
Now the estimate follows from Lemma~\ref{lem:aux-local-av}. \pbox

  Since we assume that $g,\varphi\in\S(\R^d)$, there exists $m_{g,\varphi}>0$ such that
\begin{equation}\label{eq:decay-K_g}
  |\langle  \varphi,\pi(z)g\rangle|\leq   m_{g,\varphi} \big(1+|z|\big)^{- 2(d+1)}.
\end{equation}
If $\varphi=g$, we write $m_\varphi$.

 \begin{theorem}\label{thm:error-outside}
Let $\sigma=1$, $\delta\in[\delta_0,1],\ \delta_0>0,$ and $\mathcal{A}(\delta)$ as in Lemma~\ref{lem:sup}.
 Then there exist constants $c,C>0$ depending only on $\varphi$ and $\delta_0$ such that the probability of an estimation error occurring in $\Omega^c\cap  \mathcal{A}(\delta)$ satisfies
\begin{equation}\label{eq:thm-error-outside}
\P\left(\sup_{z\in \Omega^c\cap \mathcal{A}(\delta)}\rho(z)\geq \delta\right)\leq c \cdot|\Omega| \cdot e^{-C K }.
\end{equation}
\end{theorem}

\noindent  \proof
By the Besicovitch covering theorem \cite{MR12325}, there exist points $\{z_k\}_k\subset \mathcal{A}(\delta) $  such that 
  $ \mathcal{A}(\delta)  \subset \bigcup_kB_{ 1 }(z_k)$, and 
\begin{equation}\label{eq:besico}\sup_{z\in \R^{2d}} \# \big\{k:\ z\in B_{1}(z_k)\big\}\leq b_{d}<\infty,
\end{equation}
  where $b_d$ is a universal constant depending only on the dimension $d$.

Next, we show that $\sup_{w\in\R^{2d}}\sum_k\mathcal{K}(z_k,w)$ is bounded by a constant only depending on $\varphi$.  To see this, we first note that if $z\in B_{1}(z_k)$ and $|z_k-w|\geq n-1$, then $|z-w|\geq n-2$, and,
$$
\bigcup_{ n-1\leq |z_k-w|<n} B_{1}(z_k)\subset B_{n+ {1} }(w)\backslash B_{n-  {2}}(w).
$$ 
Hence, by \eqref{eq:besico}, there are at most 
$b_d\big(|B_{n+1}(0)|-|B_{n- {2}}(0)|\big)/B_1(0)$ such points $z_k$. 
 Consequently,  \eqref{eq:decay-K_g} implies
  \begin{align*}
  \sum_k\mathcal{K}(z_k,w)&\leq m_\varphi\sum_{n\in\N}\sum_{ { n-1\leq |z_k-w|<n}}\sup_{z\in B_{1 }(z_k)}(1+|z-w| )^{-2(d+1)}
  \\
  &\lesssim m_\varphi b_d\left(1+\sum_{n\geq 2} (1+ (n-2))^{-2(d+1)} \big({|B_{n+ {1} }(0)|-|B_{n- {2}}(0)|}\big) \right)
  \\  &\lesssim m_\varphi b_d\left(1+ \sum_{n\geq 2}  n^{-2(d+1)} n^{2d}\right)\lesssim m_\varphi< \infty.
  \end{align*}  
    Taking the union bound, and    applying Lemma~\ref{lem:sup}    then shows
  \begin{align*}
  \P\left(\sup_{z\in  \mathcal{A}(\delta) \cap \Omega^c}\rho(z)\geq \delta\right)&\leq \sum_k\P\left(\sup_{z\in B_{1 }(z_k)}\rho(z)\geq \delta\right)
  \\
  & \leq c \cdot e^{-C K} \sum_k \int_{\R^{2d}}\vartheta(z) \mathcal{K} (z_k,z){\ud z} 
  \\
    & \leq c  \cdot e^{-C  K}\sup_{w\in\R^{2d}}\sum_k\mathcal{K}(z_k,w)\int_{\R^{2d}}   \vartheta(z) {\ud z} 
    \\
    & \lesssim m_\varphi c \cdot e^{-C  K} |\Omega|,
  \end{align*} 
  where we used $\|\vartheta\|_1\leq |\Omega|$ (see Proposition~\ref{prop:estimates}~\ref{enum:3}).
  \pbox

  Next we estimate the size of the region $\Omega^c\backslash \mathcal{A}(\delta)$.  
 
 \begin{prop}\label{prop:size-Ag^c}
Let $\delta>0$. Then there exists a constant $c>0$ only depending on $\varphi$ such that 
\begin{equation}\label{eq:Ag-size}
\left| \Omega^c\backslash \mathcal{A}(\delta)\right|\leq c\cdot M({g,\varphi})\cdot\frac{|\partial \Omega|}{\delta},
\end{equation}
and
$$
\int_{\R^{2d}}\vartheta(w)\mathcal{K}(z,w)\ud w\leq c \cdot M({g,\varphi}) \cdot\emph{dist}(z,\partial\Omega)^{-1},
\qquad z\in\Omega^c.
$$
 \end{prop}
\noindent \proof Observe that 
$$
\mathcal{K}(z,w)\leq  m_\varphi \sup_{y\in B_1(z)} (1+|y-w| )^{- 2(d+1)} \lesssim m_\varphi  (1+|z-w| )^{-  2(d+1)}.
$$ 
Then we may estimate the following integral
\begin{align*}
\int_{\Omega^c}\int_{\R^{2d}} \vartheta(w)&\mathcal{K}(z,w) \ud w \ud z =\int_{\Omega^c}\int_{\Omega} \vartheta(w)\mathcal{K}(z,w)\ud w \ud z +\int_{\Omega^c}\int_{\Omega^c} \vartheta(w)\mathcal{K}(z,w)\ud w \ud z\nonumber
\\
&\lesssim m_\varphi\int_{\Omega^c}\int_{\Omega} (1+|z-w| )^{-2(d+1)}\ud w \ud z +\hspace{-0.03cm}\int_{\Omega^c} \hspace{-0.03cm}\vartheta(w)\hspace{-0.03cm}\int_{\Omega^c}\hspace{-0.03cm}\mathcal{K}(z,w)\ud z \ud w \nonumber
\\ 
&\lesssim  m_\varphi\cdot |\partial\Omega|+m_\varphi\cdot M({g,\varphi}) \cdot |\partial\Omega|,\label{eq:om-back-ag}
\end{align*}
where we applied \eqref{eq_reg} with $\psi(z)=(1+|z| )^{-2(d+1)}$ and Proposition~\ref{prop:estimates}~\ref{enum:1} to derive the estimates on the first and second integral respectively. Then \eqref{eq:Ag-size} follows from Chebyshev's inequality.

To show the second property, we set $R=\text{dist}(z,\partial\Omega)$ and assume that $z\in \Omega^c$.  By Proposition~\ref{prop:estimates}~\ref{enum:4} it holds
\begin{align*}
\int_{\R^{2d}}\vartheta(w)\mathcal{K}(z,w)\ud w&= \int_{\Omega+B_{R/2}(0)}\vartheta(w)\mathcal{K}(z,w)\ud w+ \int_{(\Omega+B_{R/2}(0))^c}\vartheta(w)\mathcal{K}(z,w)\ud w
\\
&\lesssim \int_{\Omega+B_{R/2}(0)} \mathcal{K}(z,w)\ud w+ M({g,\varphi}) (R/2)^{-1} \int_{(\Omega+B_{R/2}(0))^c} \mathcal{K}(z,w)\ud w
\\
&\leq m_\varphi  \int_{\Omega+B_{R/2}(0)} (1+|z-w|)^{-2(d+1)}\ud w +  M({g,\varphi}) R^{-1} 
\|\mathcal{K}(z,\cdot)\|_1\\
&\leq m_\varphi  \int_{ B_{R/2}(0)^c} (1+|w|)^{-2(d+1)}\ud w +  M({g,\varphi}) R^{-1} 
\|\mathcal{K}(0,\cdot)\|_1\\
&\lesssim \big(m_\varphi   + C_\varphi M({g,\varphi}) \big)R^{-1}.
\end{align*}
\pbox
 
  We will now consider  estimates for the uniform error probability in $\Omega$. We first bound the probability that an estimation error occurs in a set of the form $\Omega\cap \{ |\vartheta -1|\leq  \delta\}$.

  \begin{theorem}\label{thm:error-inside}
Let $\sigma=1$ and $\delta\in(0,1/2]$.  There exist  constants $c,C>0$  only depending on $\varphi$ such that
\begin{equation}
\P\left(\inf_{z\in\Omega\hspace{0.05cm}\cap\hspace{0.05cm} \{ |\vartheta -1|\leq \delta/3\}} \rho(z) < \delta\right) \leq c\cdot|\Omega^1| \cdot e^{-C {K}}.
\end{equation}
\end{theorem}

\noindent
\proof
First, we observe that
\begin{align}\label{eq:inf-mean-value}
\begin{aligned}
\inf_{z\in B_R(z_0)}\rho(z)&\geq \rho(z_0)-\sup_{z\in B_R(z_0)}|\rho(z_0)-\rho(z)|
\\
&\geq \rho(z_0)-\sup_{z\in B_R(z_0)}\frac{1}{K}\sum_{k=1}^K \left||V_{\varphi}H_\Omega \No_k(z_0)|^2-|V_{\varphi}H_\Omega \No_k(z)|^2\right|.
\end{aligned}
\end{align}
 The reproducing formula \eqref{eq:rep} and Cauchy-Schwarz inequality then  show that for any $f\in L^2(\R)$, $0<R\leq 1$, and $z\in B_R(z_0)$
\begin{align*}
\Big||V_{\varphi}f(z)|^2-&|V_{\varphi}f(z_0)|^2\Big|  =\big(|V_{\varphi}f(z)|-|V_{\varphi}f(z_0)|\big)\big(|V_{\varphi}f(z)|+|V_{\varphi}f(z_0)|\big)
\\
& = \big||V_{\varphi}\pi(-z_0)f(z-z_0)|-|V_{\varphi}\pi(-z_0)f(0)|\big|\big(|V_{\varphi}f(z)|+|V_{\varphi}f(z_0)|\big)
\\
& \leq  \int_{\R^{2d}} |V_{\varphi}\pi(-z_0)f(w)\langle \pi(z-z_0)\varphi-\varphi,\pi(w)\varphi\rangle|\ud w  
\\
& \hspace{1.2cm}\cdot \int_{\R^{2d}} |V_{\varphi}f(w)|\big(|{K}_\varphi(z,w)|+|{K}_\varphi(z_0,w)|\big)\ud w
\\
& \leq    \int_{\R^{2d}} |V_{\varphi}\pi(-z_0)f(w)\langle \pi(z-z_0)\varphi-\varphi,\pi(w)\varphi\rangle|\ud w  \cdot 2
 \int_{\R^{2d}} |V_{\varphi}f(w)| \mathcal{K}(z_0,w) \ud w
\\
& \leq  2C_\varphi^{1/2} \left(\int_{\R^{2d}} |V_{\varphi}\pi(-z_0)f(w)|^2|\langle \pi(z\hspace{-1pt}-\hspace{-1pt}z_0)\varphi-\varphi,\pi(w)\varphi\rangle|\ud w  \right)^{\frac{1}{2}}
\\
&\hspace{1.2cm}\cdot\left(\int_{\R^{2d}} |\langle \pi(z-z_0)\varphi-\varphi,\pi(w)\varphi\rangle | \ud w\right)^{\frac{1}{2}} \cdot \left(\int_{\R^{2d}} |V_{\varphi}f(w)|^2\mathcal{K}(z_0,w)\ud w\right)^{\frac{1}{2}}
\\
& \leq   2^{3/2}C_\varphi^{1/2}  \int_{\R^{2d}}|V_{\varphi}f(w)|^2\mathcal{K}(z_0,w)\ud w \cdot\left(\int_{\R^{2d}} |\langle \pi(z-z_0)\varphi-\varphi,\pi(w)\varphi\rangle | \ud w\right)^{\frac{1}{2}}
\\
& \leq  T_\varphi(R)\int_{\R^{2d}}|V_{\varphi}f(w)|^2\mathcal{K}(z_0,w)\ud w,
\end{align*}
where 
$$
T_\varphi(R):=  2^{3/2}C_\varphi^{1/2} \sup_{z\in B_R(0)}\left(\int_{\R^{2d}} |\langle \pi(z)\varphi-\varphi,\pi(w)\varphi\rangle | \ud w\right)^{\frac{1}{2}}\rightarrow 0
,\quad \text{as }R\rightarrow 0,
$$
because $\varphi$ is a Schwartz function. Therefore, 
\begin{equation}\label{eq:sup-osc}
\sup_{z\in B_R (z_0)}|\rho(z_0)-\rho(z)|\leq T_\varphi(R)\int_{\R^{2d}}\rho(w)\mathcal{K}(z_0,w)\ud w.
\end{equation}
Let now $z_0\in \big\{z\in\R^{2d}:\ |\vartheta(z)-\chi_\Omega(z)|\leq \delta/3\big\}\cap \Omega$. Then, combining the assumption $\delta\in(0,1/2]$ with \eqref{eq:inf-mean-value} and \eqref{eq:sup-osc} yields
\begin{align*}
\P \left(\inf_{z\in B_R(z_0)}\rho(z)< \delta\right)&
  \leq \P \left( \rho(z_0)-  \vartheta(z_0)< \delta-  \vartheta(z_0)+T_\varphi(R)\int_{\R^{2d}} \rho(w) \mathcal{K} (z_0,w)\ud w \right)
\\
& \leq \P\left(\rho(z_0)- \vartheta(z_0)< -\frac{1}{3}  +T_\varphi(R)\int_{\R^{2d}} \rho(w) \mathcal{K} (z_0,w)\ud w\right).
\end{align*}
In the following, we make use of   $\P(A)=\P(A\cap B)+\P(A\cap B^c)\leq \P(A\cap B)+\P(B^c)$. Define the events 
$$
A_{z_0}:=\left\{\rho(z_0)-  \vartheta(z_0)< -\frac{1}{3} +T_\varphi(R)\int_{\R^{2d}}\rho(w) \mathcal{K} (z_0,w)\ud w\right\},
$$
and
\begin{equation}\label{eq:Bz0}
B_{z_0}:=\left\{T_\varphi(R)\int_{\R^{2d}} \rho(w) \mathcal{K} (z_0,w)\ud w< \frac{1}{6}\right\}.
\end{equation}
Then, we may   estimate
 \begin{align*}
 \P(A_{z_0}\cap B_{z_0})\leq \P\left(\rho(z_0)- \vartheta(z_0)< -\frac{1}{6}  \right)\leq \P\left(|\rho(z_0)- \vartheta(z_0)|>  \frac{1}{6}  \right)\lesssim e^{-C {K}},
  \end{align*}
  where we used Lemma~\ref{lem:aux-hw} and   $\vartheta(z_0)\leq 1$. 
  
  Second, choose $R_\varphi \leq 1$ so that
   \begin{equation}\label{def:Rg}
    \frac{1}{6 T_\varphi( 
   R_\varphi)}\geq 2 C_\varphi,
   \end{equation} 
   where $C_\varphi$ is the constant defined in \eqref{eq:Cphi}.
  Then we estimate $\P({B}_{z_0}^c)$ as follows: 
  \begin{align*}
  \P(B_{z_0}^c)&=\P\left(  \int_{\R^{2d}} \rho(w) \mathcal{K} (z_0,w)\ud w\geq  \frac{1}{6T_\varphi(R_\varphi)}   \right)
  \\
  &\leq \P\left(   \int_{\R^{2d}} |\rho(w)- \vartheta(w)|\mathcal{K}(z_0,w)\ud w  \geq  \frac{1}{6T_\varphi(R_\varphi)}-  \int_{\R^{2d}}\hspace{-0.12cm}\vartheta(w)\mathcal{K}(z_0,w)\ud w  \right)
    \\
  &\leq \P\left(   \int_{\R^{2d}}|\rho(w)- \vartheta(w)|\mathcal{K} (z_0,w){\ud w} \geq  \frac{1}{6T_\varphi(R_\varphi )}- C_\varphi \right)
      \\
  &\leq \P\left(  \int_{\R^{2d}}|\rho(w)- \vartheta(w)|\mathcal{K} (z_0,w){\ud w} \geq  C_\varphi  \right)
  \\
    &\leq \P\left(  \int_{\R^{2d}}|\rho(w)- \vartheta(w)|\mathcal{K} (z_0,w){\ud w} \geq  1 \right)
  \\
  &\leq c\cdot e^{-CK}\int_{\R^{2d}}\vartheta(w)\mathcal{K}(z_0,w) {\ud w} 
 \lesssim m_\varphi \cdot c\cdot e^{-C {K}},
    \end{align*}
where we used Lemma~\ref{lem:aux-local-av}.

Let us consider a covering of $ \Omega\cap \{z\in\R^{2d}:\ |\vartheta(z)-1|\leq \delta/3\}$ 
by balls $B_{R_\varphi}(z_k)$, $z_k\in \Omega\cap \{ |\vartheta-1|\leq \delta/3\}$, with at most  $b_{d}$ overlaps. Then there are at most  $|\Omega^{R_\varphi}|b_d/(\pi R_\varphi^2)$ 
points $z_k$ defining such a  covering. A union bound argument then yields
\begin{align*}
\P\left(\inf_{z\in\Omega\cap \{ |\vartheta-1|< \delta/3\}} \rho(z) \leq \delta\right)&\leq \sum_{k}\P\left(\inf_{z\in B_{R_\varphi}(z_k)} \rho(z)< \delta\right)
\\
& \leq 
\sum_{k} P(A_{z_k}\cap B_{z_k}) + P( B_{z_k}^c)
\\
&\lesssim c\sum_{k}  m_\varphi e^{-C{K} } 
\lesssim c\frac{ m_\varphi}{R_\varphi^2 }  | \Omega^{R_\varphi}| e^{-C {K} }\leq c\frac{ m_\varphi}{R_\varphi^2 }  | \Omega^1| e^{-C {K} },
\end{align*}
as $R_\varphi\leq 1$.
\pbox

  So far we have considered a generic deterministic cutoff parameter $\delta$. The following result will allow us to replace $\delta$ by the random threshold $\|\rho\|_\infty/4$.
\begin{prop}\label{prop:max-rho}
Let $\sigma=1$, and assume \eqref{eq_as2}. Then there exist    constants $c,C>0$ depending only on $\varphi$ and a constant $r>0$ depending on $g$ and $\varphi$ such that  
$$
\P(  \|\rho\|_\infty\leq 2)\geq 1-c\cdot  | \Omega^r|\cdot e^{-C K}.
$$
In addition, if $\delta\in [0,(20\|V_\varphi g\|_1^2)^{-1} ]$, and $K\geq 2$, then 
$$
\P(\|\rho\|_\infty \leq \delta)\leq (20\|V_\varphi g\|_1^2\delta)^{K/2}.
$$
\end{prop}
\noindent \proof 
For the event $\{ \|\rho\|_\infty>2\}$, we   argue as in the proof of Theorem~\ref{thm:error-inside}. With the notation of that proof,
\begin{align*}
\sup_{z\in B_R(z_0)}\rho(z)&\leq \rho(z_0)+\sup_{z\in B_R(z_0)}|\rho(z_0)-\rho(z)|
\\
&\leq \rho(z_0)+T_\varphi(R)\int_{\R^{2d}}\rho(w)\mathcal{K}(z_0,w)\ud w.
\end{align*}
Therefore, since $2\geq 1+\vartheta(z_0)$
\begin{align*}
\P\left(\sup_{z\in B_R(z_0)}\rho(z)>2\right)&\leq \P\left(\rho(z_0)>1+\vartheta(z_0)-T_\varphi(R)\int_{\R^{2d}}\rho(w)\mathcal{K}(z_0,w)\ud w\right)
\\
&\leq \P(B_{z_0}^c)+\P\big(|\rho(z_0)-\vartheta(z_0)|>5/6\big),
\end{align*}
where $B_{z_0}$ is the event given by \eqref{eq:Bz0}. Proposition~\ref{prop:size-Ag^c} guarantees the existence of a constant $r=r({g,\varphi})>0$  such that
$(\Omega^c\backslash  \mathcal{A}(1)) \cup \Omega\subset  \Omega^r$. Therefore, 
\begin{align*}
\P(\|\rho\|_\infty>2) 
&\leq \P\left(\sup_{z\in\Omega^c\cap \mathcal{A}(1)}\rho(z)>1\right)+\P\left(\sup_{z\in(\Omega^c\backslash  \mathcal{A}(1)) \cup \Omega}\rho(z)>2\right)
\\
&\leq c|\Omega|e^{-C K}+\P\left(\sup_{z\in  \Omega^r }\rho(z)>2\right),
\end{align*}
by  Theorem~\ref{thm:error-outside}. For the second term, we cover $\Omega^r$ with balls  of radius $R_\varphi\leq 1$ satisfying \eqref{def:Rg} and apply a union bound argument as in the proof of Theorem~\ref{thm:error-inside} to show 
$$
\P\left(\sup_{z\in   \Omega^r}\rho(z)>2\right)\lesssim \sum_k \Big(\P( {B}_{z_k}^c)+\P\big(|\rho(z_k)-\vartheta(z_k)|>5/6\big) \Big) \lesssim \frac{m_\varphi}{R_\varphi^2} | \Omega^{r+1}|e^{-CK}.
$$
To show the second property we first observe that, by the reproducing formula \eqref{eq:rep}, the averaged spectrogram formed with respect to the model window $g$ is dominated by $\rho$ (which involves the reconstruction window $\varphi$)
\begin{align*}
 \sup_{z\in\R^{2d}}&\frac{1}{K}\sum_{k=1}^K |V_g H_\Omega \No_k(z)|^2 = \sup_{z\in\R^{2d}}
 \frac{1}{K }\sum_{k=1}^K \left|\int_{\R^{2d}}V_\varphi H_\Omega \No_k(w)\langle \pi(w)\varphi,\pi(z)g\rangle \ud w\right|^2
 \\
 &\leq \sup_{z\in\R^{2d}}
  \int_{\R^{2d}} \frac{1}{K }\sum_{k=1}^K\big|V_\varphi H_\Omega \No_k(w)\big|^2|\langle \pi(w)\varphi,\pi(z)g\rangle| \ud w \int_{\R^{2d}}|\langle \pi(w)\varphi,\pi(z)g\rangle| \ud w
  \\
 &\leq \sup_{z\in\R^{2d}}
  \|\rho\|_\infty \left( \int_{\R^{2d}}|\langle \pi(w)\varphi,\pi(z)g\rangle| \ud w\right)^2=  \|V_\varphi g\|_1^2  \|\rho\|_\infty.
\end{align*}
Next, we
 use the double orthogonality property \eqref{eq:double-orth} of the eigenfunctions $f_n$ to derive
\begin{align*}
\int_\Omega \sum_{k=1}^K |V_gH_\Omega \No_k(z)|^2\ud z&= \sum_{k=1}^K\sum_{\ell,n\in\N}\alpha_\ell^k\overline{\alpha_n^k}\lambda_\ell\lambda_n \int_\Omega V_gf_\ell(z) \overline{V_gf_n(z)}\ud z
\\
&= \sum_{k=1}^K\sum_{n\in\N}|\alpha_n^k|^2\lambda_n^3.
\end{align*}
Therefore, there exist $z_0\in\Omega$ such that 
$$
 \frac{1}{K}\sum_{k=1}^K |V_gH_\Omega \No_k(z_0)|^2\geq \frac{1}{K|\Omega|}\sum_{k=1}^K\sum_{n\in\N}|\alpha_n^k|^2\lambda_n^3\geq  \frac{1}{K|\Omega|}\left(\frac{3}{4}\right)^3\sum_{k=1}^K\sum_{n=1}^{\lfloor |\Omega|/2\rfloor}|\alpha_n^k|^2,
$$
where we used Lemma~\ref{lem:aux-eigenvalues} to derive the last inequality.
We define  $Y=\sum_{k=1}^K\sum_{n=1}^{\lfloor |\Omega|/2\rfloor}|\alpha_n^k|^2 =\sum_{k=1}^K\sum_{n=1}^{\lfloor |\Omega|/2\rfloor}(\text{Re}(\alpha_n^k)^2+\text{Im}(\alpha_n^k)^2)$. Then $2Y\sim  \chi^2(2 K\lfloor|\Omega|/2\rfloor)$ and 
\begin{align*}
\P(\|\rho\|_\infty \leq \delta)&\leq \P\left(\sup_{z\in\R^{2d}}\frac{1}{K}\sum_{k=1}^K |V_gH_\Omega \No_k(z)|^2 \leq \|V_\varphi g\|_1^2 \delta\right)
\\
&\leq \P\left(2Y\leq 2\left(\frac{4}{3}\right)^3\|V_\varphi g\|_1^2  K|\Omega|\delta\right)
\\
&= {\Gamma\left(  K\lfloor|\Omega|/2\rfloor \right)}^{-1}\gamma\left(  K\lfloor |\Omega|/2\rfloor, \left(\frac{4}{3}\right)^3\|V_\varphi g\|_1^2  K|\Omega|\delta  \right)
\\
&\leq  {\Gamma\left(  K\lfloor|\Omega|/2\rfloor \right)}^{-1}\gamma\left(  K\lfloor |\Omega|/2\rfloor, 10\|V_\varphi g\|_1^2  K\lfloor |\Omega|/2\rfloor\delta  \right),
\end{align*}
where $\gamma(a,z)$ denotes the lower incomplete gamma function and we used $|\Omega| \leq 4 \lfloor|\Omega|/2\rfloor$, 
by \eqref{eq_as2},
as well as $4(4/3)^3\leq 10$.  Subsequently, we will show that 
\begin{equation}\label{eq:est-gamma}
\gamma\left(a, ax \right)\leq (2x )^{a/2} \Gamma\left(a\right),\quad x\in[0,1/2],\ a\geq 2,
\end{equation}
which concludes our proof by setting $x=10\|V_\varphi g\|_1^2\delta$ and $a=K\lfloor |\Omega|/2\rfloor\geq 2$.

To this end, we we prove that  $h_1(x):= x^{-a/2}\gamma(a,ax)$ is non-decreasing on $(0,1/2]$; then \eqref{eq:est-gamma} will follow by
$$
x^{-a/2}\gamma\left(a, ax \right)=h_1(x)\leq h_1(1/2)= 2^{a/2}\gamma\left(a, \frac{a}{2} \right) \leq 2^{a/2} \Gamma(a).
$$
The derivative of $h_1$ is given by 
$$
h'_1(x)=x^{-a/2-1}\left((ax)^ae^{-ax}-\frac{a}{2}\gamma(a,ax)\right)=x^{-a/2-1}h_2(x).
$$
We note that $h_2(0)=0,$ and that 
$$
h_2'(x)=a(xa)^ae^{-ax}\big((2x)^{-1}-1\big)\geq 0,\quad x\in(0,1/2].
$$
This implies that $h_2(x)\geq 0$, and consequently $h_1'(x)\geq 0$ whenever $x\in(0,1/2]$.
\pbox

  We can now prove the main results.

\subsection{Proof of Theorem~\ref{mthm1}}\label{sec_5.3}
Let $\sigma=1$. We apply Theorem~\ref{thm:error-inside} and Proposition~\ref{prop:max-rho} to show
\begin{align}\label{eq_a}
\P\left(\inf_{z\in\Omega\cap\{|\vartheta-1|\leq 1/6\}}\rho(z)<\|\rho\|_\infty/4\right)&\leq \P\left(\inf_{z\in\Omega\cap\{|\vartheta-1|\leq 1/6\}}\rho(z)<1/2\right)+\P\big(\|\rho\big\|_\infty >2)
\\\notag
&\leq \frac{1}{2}c | \Omega^r| e^{-C K}.
\end{align}
The probability on left hand side of \eqref{eq_a} is homogeneous with respect to $\sigma$. Hence, the same estimate holds for all $\sigma>0$.

Similarly, for $\sigma=1$, Theorem~\ref{thm:error-outside} and Proposition~\ref{prop:max-rho} yield
\begin{align*}
\P&\left(\sup_{z\in\Omega^c\cap \mathcal{A}(1/(100\|V_\varphi g\|_1^2))}\rho(z)\geq\|\rho\|_\infty/4\right)
\\ &\hspace{2cm}\leq \P\left(\sup_{z\in\Omega^c\cap \mathcal{A}(1/(100\|V_\varphi g\|_1^2))}\rho(z)\geq 1/(100\|V_\varphi g\|_1^2)\right)+\P\big(\|\rho\big\|_\infty \leq 1/(25\|V_\varphi g\|_1^2))
\\
&\hspace{2cm}\leq \frac{1}{2} c |\Omega^{r}| e^{-C K},
\end{align*}
where $r=r({g,\varphi})$.
Again by homogeneity, the probability bound holds also for general $\sigma>0$.

Hence, up to an event of probability at most $c |\Omega^{r}| e^{-C K}$, the estimation error $\Omega \triangle \widehat{\Omega}$ is contained in the union of the sets
\begin{align}\label{eq_b}
\begin{aligned}
&\Omega\cap\{|\vartheta-1| > 1/6\},
\\
&\Omega^c \setminus \mathcal{A}\big((100\|V_\varphi g\|_1^2)^{-1}\big)\big)
=
\left\{z\in\Omega^c:\ \int_{\R^{2d}}\vartheta(w)\mathcal{K}(z,w){\ud w} >  \frac{1}{200 \|V_\varphi \varphi\|_1 \|V_\varphi g\|_1^2} \right\}.
\end{aligned}
\end{align}
Propositions~\ref{prop:estimates} and \ref{prop:size-Ag^c}
show that each of these sets is contained in
$[\partial\Omega]^{r'}$, for an adequate constant,
$r'=r'(g,\varphi)$. This completes the proof. 
\pbox

\begin{rem}\label{rem_cons}
 {If we keep track of the constants in the results leading up to the  proof of Theorem \ref{mthm1} we see that the the various constants in the statement can be chosen as
	\[r=\alpha_1   \cdot P_1(g,\varphi)  \cdot M(g,\varphi), \qquad c= \alpha_2 \cdot P_2(g,\varphi), \qquad C=\alpha_3 /P_3(g,\varphi),\] where  $\alpha_i$ are absolute constants, and $$P_i(g,\varphi)=m_\varphi^{a_i}\cdot C_\varphi^{b_i}\cdot\|V_g\varphi\|_1^{c_i}\cdot R_\varphi^{-{d_i}},$$ for some $a_i,b_i,c_i,d_i\in\{0,1,2,3,4\}$, $i=1,2,3,$ cf. \eqref{eq_c1}, \eqref{eq:Cphi}, \eqref{def:Rg}.}
\end{rem}

\subsection{Proof of Theorem~\ref{mcoro1}}\label{sec_5.4}
With the notation of the proof of Theorem \ref{mthm1}, with the same success probability, the estimation error $\Omega \triangle \widehat{\Omega}$ is contained in the union of the sets \eqref{eq_b}. In that event, we use Propositions~\ref{prop:estimates} and \ref{prop:size-Ag^c} to estimate
\begin{align*}
|\Omega\Delta \widehat{\Omega}| \leq  \big|  \{|\chi_\Omega-\vartheta|\hspace{-0.5pt}>\hspace{-0.5pt}1/6 \}\big|\hspace{-0.5pt}+\hspace{-0.5pt}\big|\Omega^c\backslash \mathcal{A}\big(   1/(100\|V_\varphi g\|_1^2)\big)\big|\leq C_1 |\partial \Omega|,
\end{align*}
with $C_1=C_1(g,\varphi)$.
\pbox

 \subsection{Proof of Theorem~\ref{mthm2}} \label{sec_5.5}

The following is a more precise version of Theorem~\ref{mthm2}.
 
\begin{theorem}\label{thm:error-expec}
Let $g,\varphi\in\S(\R^d)$ with $\norm{2}{g}=\norm{2}{\varphi}=1$, $K\geq 7$, and $\Omega\subset\R^{2d}$ compact and satisfying \eqref{eq_as2}. Then there exist  constants $c_1,c_2,C,r>0$ depending only on $g$ and $\varphi$ such that
\begin{equation}\label{eq:thm-exp}
\E\left\{ |\Omega\Delta \widehat{\Omega}|\right\}\leq   c_1\cdot |\partial\Omega| +   c_2\cdot| \Omega^r|^{2} \cdot e^{-C K }. 
\end{equation}
Consequently, choosing $K\geq  C'\log\left(  | \Omega^r|       \right)$ yields
$$
\E\left\{|\Omega\Delta \widehat{\Omega}|\right\}\leq c_3\cdot|\partial\Omega|,
$$
with $c_3$ depending only on $g$ and $\varphi$.
\end{theorem}
\noindent\proof Let us assume that  $\sigma=1$. The general case then follows again by homogeneity.
We define the family of events $E_0=\{4\leq\|\rho\|_\infty<\infty\}$, and $E_{n}=\{a_n\leq  \|\rho\|_\infty<b_n\}$, with $a_n=(20\|V_\varphi g\|_1^2 n)^{-1}$, $b_n=(20\|V_\varphi g\|_1^2(n-1))^{-1}$ for $n\geq 2$,
and $a_1=1/(20\|V_\varphi g\|_1^2)$, $b_1= 4$.
Note that $\bigcup_{n \geq 0} E_n = \{ \|\rho\|_\infty \in (0,\infty)\}$ and the complementary event has zero probability. Hence,
\begin{align}\label{eq:exp}
\E\big\{|\Omega\Delta \widehat{\Omega}|\big\}=\E\left\{\big\|\chi_\Omega-\chi_{S_{\|\rho\|_\infty/4}}\big\|_{1}\right\}
&=   \int_{\R^{2d}}\E\left\{\big|\chi_\Omega(z)-\chi_{S_{\|\rho\|_\infty/4}}(z)\big|\right\}\ud z\notag
\\
&=\sum_{n\in\N_0} \int_{\R^{2d}}\P\left(\big\{ \chi_\Omega(z)\neq \chi_{S_{\|\rho\|_\infty/4}}(z)\big\}  \cap E_n\right)\ud z.
\end{align}
For $n\geq 1$,  we first observe that
\begin{equation}\label{eq:anbn}
\P\left(\big\{ \chi_\Omega(z)\neq \chi_{S_{\|\rho\|_\infty/4}}(z)  \big\}  \cap E_n\right)\leq \left\{\begin{array}{ll} \P\left( \chi_\Omega(z)\neq \chi_{S_{b_n/4}}(z)   \right),  & \text{if } z\in\Omega \\  \P\left( \chi_\Omega(z)\neq \chi_{S_{a_n/4}}(z)   \right),  &  \text{if } z\in\Omega^c
\end{array} \right. .
\end{equation}
Using Propositions \ref{prop:estimates} and \ref{prop:max-rho} we can estimate for $n\geq 3$ and $K\geq 7$
\begin{align}
\int_{\R^{2d}}   \P\big(\big\{&\chi_\Omega(z) \neq\chi_{S_{\|\rho\|_\infty/4}}(z)\big\}\cap E_n\big)\ud z \notag \\ 
& \leq   \int_{\Omega\cup  \{|\vartheta-\chi_\Omega|> a_n^2 \}} \P(E_n)\ud z +\int_{\Omega^c\cap\{|\vartheta-\chi_\Omega|\leq a_n^2\}}  \P\big(\chi_\Omega(z)\neq\chi_{S_{ a_n/4}}(z)\big)\ud z \notag
\\
&\lesssim (n-1)^{-K/2} \left(|\Omega|+n^2 \|V_\varphi g\|_1^4   M({g,\varphi}) |\partial\Omega|\right)+\int_{\Omega^c\cap\{ \vartheta \leq a_n^2\}}  \P\big(\chi_\Omega(z)\neq\chi_{S_{ a_n/4}}(z)\big)\ud z \notag
\\
&\lesssim \left(\frac{1}{2}\right)^{K/4}  {n}^{-7/4}  |\Omega|+  {n^{-3/2}}  |\partial\Omega| +\int_{\Omega^c\cap\{ \vartheta \leq a_n^2\}}  \P\big(\chi_\Omega(z)\neq\chi_{S_{ a_n/4}}(z)\big)\ud z. \label{eq:En}
\end{align}
If $\vartheta(z)\leq \delta^2$, then $\text{exp}\left({-C   { K }/{\delta}}\right)\geq \text{exp}\left({-C  { K }/{\sqrt{\vartheta(z)}}}\right)\geq \text{exp}\left({-C { K\delta}/{\vartheta(z)}}\right)$. 
Therefore, Theorem~\ref{thm:high-prob} allows us to further estimate the integral  above for a general parameter $\delta>0$ as follows
\begin{align*}
\int_{\Omega^c\cap\{\vartheta \leq \delta^2\}}  \P\big(\chi_\Omega(z) \neq\chi_{S_{\delta/4}}(z)\big)\ud z&\lesssim \int_{\Omega^c\cap \{ \vartheta \leq\delta^2\}} e^{-C \frac{K\delta}{ \vartheta(z)}}\ud z
\\
&= \int_0^\infty \left|\left\{z\in\left\{ \vartheta \leq{\delta}^2\right\}:\ \text{exp}\left({-C  \frac{ K\delta}{\vartheta(z)}}\right) \geq t\right\}\right|\ud t
\\
&\leq  \int_0^\infty \left|\left\{z\in\left\{ \vartheta \leq{\delta}^2\right\}:\ \text{exp}\left({-C  \frac{ K}{\sqrt{\vartheta(z)}}}\right) \geq t\right\}\right|\ud t
\\
&\leq  \int_0^{e^{-C  K/\delta }} \left|\left\{z\in\left\{ \vartheta \leq{\delta}^2\right\}:\   \sqrt{\vartheta(z)}\geq \frac{C  K}{\log(1/t)}\right\}\right|\ud t
\\
&\lesssim \int_0^{e^{-C  K/\delta }}  \frac{\log(1/t)}{K }\int_{\R^{2d}} \sqrt{\vartheta(z)} \ud z \ud t
 \\
 & = \frac{ \big\|\sqrt{\vartheta}\big\|_{1}}{K } \left(- x\log x+x\right)\big|_0^{e^{-CK/\delta}}\lesssim {\big\|\sqrt{\vartheta}\big\|_{1}} e^{-C  K/\delta}.
\end{align*}
In order to estimate $\big\|\sqrt{\vartheta}\big\|_1$ we first note that the integral kernel of $H_\Omega$ can be written in two ways if we either use its eigenexpansion or \eqref{eq_locop}
$$
k_{H_\Omega}(s,t)=\sum_{n\in\N} \lambda_n \overline{f_n(t)}f_n(s)=\int_\Omega \overline{\pi(w)g(t)}\pi(w)g(s)\ud w.
$$
Integrating against $\pi(z) \varphi(t)$ and $\overline{\pi(z) \varphi(s)}$ thus yields
$$
\vartheta(z)\leq \sum_{n\in\N} \lambda_n |V_\varphi f_n(z)|^2=\int_\Omega |\langle \pi(z)\varphi,\pi(w)g\rangle|^2\ud w.
$$
 To estimate $\big\|\sqrt{\vartheta}\big\|_1$, let $Q_k=k+[-\frac{1}{2},\frac{1}{2}]^{2d}$, $k\in \Z^{2d}$, and $I_\Omega=\{k\in\Z^{2d}:\ |Q_k\cap \Omega|>0.\}$. Using \eqref{eq:decay-K_g}  shows 
\begin{align}
\big\|\sqrt{\vartheta}\big\|_1&\leq \int_{\R^{2d}}\left(\int_\Omega |\langle \pi(z)\varphi,\pi(w)g\rangle|^2\ud w\right)^{1/2}\ud z \notag
\\
&\leq  \sum_{k\in\Z^{2d}}\int_{\R^{2d}}\left(\int_{\Omega\cap Q_k} |\langle \pi(z)\varphi,\pi(w)g\rangle|^2\ud w\right)^{1/2}\ud z \notag
\\
&\leq  \sum_{k\in I_\Omega} |\Omega\cap Q_k|^{1/2}\int_{\R^{2d}}\sup_{w\in Q_k} |\langle \pi(z-w)\varphi, g\rangle| \ud z \notag
\\
&\leq m_{g,\varphi} \sum_{k\in I_\Omega}  \int_{\R^{2d}}\sup_{w\in Q_0}(1+|z-w|^2)^{- (d+1)} \ud z\notag
\\
&\lesssim  |I_\Omega| \leq  \big|\Omega^{d^{1/2}}\big|
 \label{eq:sqrt-theta},
\end{align}
where we used that $Q_k\subset \Omega^{d^{1/2}}$ for every $k\in I_k$.

In the case $n=0$, we use  Propositions~\ref{prop:estimates}~\ref{enum:3} and \ref{prop:max-rho} and  argue similarly as in the inequalities leading up to \eqref{eq:En}  
\begin{align}
\int_{\R^{2d}}\P\Big( \big\{\chi_\Omega(z)&\neq \chi_{S_{\|\rho\|_\infty/4}}(z)  \big\}\cap E_0\Big)\ud z \notag
\\ 
&\lesssim  \int_{\Omega\cup \{|\vartheta-\chi_\Omega|> \frac{1}{4}\}}\P\big(E_0\big)\ud z+\int_{\Omega^c\cap \{|\vartheta-\chi_\Omega|\leq \frac{1}{4}\}}\P\big(\chi_\Omega(z)\neq\chi_{S_{ {1}}}(z)\big)\ud z\notag
\\ 
&\lesssim c'\big(| \Omega| +  |\partial\Omega|\big)\cdot| \Omega^r|e^{-C K} +\int_{\Omega^c\cap \{|\vartheta-\chi_\Omega|\leq \frac{1}{4}\}}\P\big(\chi_\Omega(z)\neq\chi_{S_{ {1}}}(z)\big)\ud z\notag
\\ 
&\leq c''| \Omega^r |^2 e^{-C K} +\int_{ \Omega^c\cap \{|\vartheta-\chi_\Omega|\leq \frac{1}{4}\}}\P\big(\chi_\Omega(z)\neq\chi_{S_{ {1}}}(z)\big)\ud z.\label{eq:int1}
\end{align}
For $n=1,2$, we finally get from Propositions~\ref{prop:estimates}~\ref{enum:3}, \eqref{eq:anbn}, and Theorem~\ref{thm:high-prob}
\begin{align}
\int_{\R^{2d}}\P\Big( \big\{&\chi_\Omega(z) \neq \chi_{S_{\|\rho\|_\infty/4}}(z)  \big\}\cap E_n\Big)\ud z \notag
\\
&\leq\int_{\big(\Omega^c \cap \{ \vartheta > \frac{a_n}{16}\}\big)\cup\big(\Omega\cap \{|1- \vartheta| > \frac{b_n}{16}\}\big) }
1 \ud z+\int_{\Omega \cap \{ |1-\vartheta |\leq \frac{b_n}{16}\}} \P\big(\chi_\Omega(z)\neq \chi_{S_{b_n/4}}(z)  \big)\ud z \notag
\\
&\qquad\qquad  +\int_{\Omega^c\cap \{ \vartheta \leq \frac{a_n}{16}\}} \P\big(\chi_\Omega(z)\neq\chi_{S_{ {a_n/4}}}(z)\big)\ud z  \notag
\\
& \lesssim   |\partial\Omega|+|\Omega|e^{-CK}
  +\int_{\Omega^c\cap \{ \vartheta \leq \frac{a_n}{16}\}} \P\big(\chi_\Omega(z)\neq\chi_{S_{ {a_n/4}}}(z)\big)\ud z .\label{eq:int2}
\end{align}
In order to bound the integrals in \eqref{eq:int1} and \eqref{eq:int2}
we argue as before and use Theorem~\ref{thm:high-prob}, and  Proposition~\ref{prop:estimates}~\ref{enum:3} to show that 
\begin{align*}
\int_{\Omega^c\cap\{|\vartheta-\chi_\Omega|\leq \delta/4\}}  \P\big(\chi_\Omega(z) \neq\chi_{S_{\delta}}(z)\big)\ud z&\lesssim \int_{ \{|\vartheta-\chi_\Omega|<\delta/4\}} e^{-C \frac{K \delta}{ \vartheta(z)}}\ud z
\\
&\leq \int_0^{e^{-C  K \delta }} \left|\left\{z\in\R^{2d}:\   \vartheta(z)\geq \frac{C K \delta}{\log(1/t)}\right\}\right|\ud t
\\
&\lesssim \|\vartheta\|_1 \int_0^{e^{-C K \delta}}  \frac{\log(1/t)}{K\delta } \ud t
 \lesssim {|\Omega|}\delta^{-1} e^{-C  K \delta},
\end{align*}
and choose $\delta\in \{1,a_1,a_2\}$.
Combining this estimate with \eqref{eq:exp} and Proposition~\ref{prop:max-rho} finally yields
\begin{align*}
\E \left\{\big\|\chi_\Omega-\chi_{S_{\|\rho\|_\infty/4}}\big\|_{1}\right\}
   \lesssim &  \sum_{n=3}^\infty \left( e^{-CK} n^{-7/4} |\Omega|+ n^{-3/2}  |\partial\Omega|+ |\Omega^r|   e^{-C  K n } \right)
\\
&\qquad \qquad   +  |\partial\Omega|+|\Omega|e^{-C K}+   | \Omega^r|^2 e^{-C K} 
\\
 \leq &  c_1  |\partial\Omega|+c_2|\Omega^r|^2 e^{-C' K}.
\end{align*}
\pbox

\begin{rem}\label{rem_cons2}
	 {From the proof of Theorem \ref{thm:error-expec} we see that the the various constants in the statement can be chosen as  
			\[r=\alpha_1   \cdot P_1(g,\varphi)  \cdot M(g,\varphi), \qquad c_1= \alpha_2 \cdot P_2(g,\varphi) \cdot M(g,\varphi), \qquad c_2=\alpha_3 \cdot P_3(g,\varphi) \cdot M(g,\varphi), \] and $C={\alpha_4}/{P_4(g,\varphi)},$ where  $\alpha_i>0$ are absolute constants, and $$P_i(g,\varphi)=m_\varphi^{a_i}\cdot m_{g,\varphi}^{b_i}\cdot C_\varphi^{c_i}\cdot\|V_g\varphi\|_1^{d_i}\cdot R_\varphi^{-{e_i}},$$ for some $a_i,b_i,d_i,e_i\in\{0,1,2,3,4\}$, $i=1,2,3,4$, cf. \eqref{eq_c1}, \eqref{eq:Cphi}, \eqref{def:Rg}.}
\end{rem}

\end{document}